\documentclass{article}
\usepackage{url}
\usepackage{color}
\usepackage{fix-cm}
\usepackage{subfigure}
\usepackage{multirow}
\usepackage[final]{graphicx}
\usepackage{pgf}
\usepackage{tikz}
\usepackage{amssymb,amsmath,amsthm}
\usepackage{authblk}
\usepackage{algorithm}
\usepackage{algorithmic}
\usepackage{appendix}
\usepackage{bm}
\usepackage{cite}
\usepackage{enumerate}

\newenvironment{definition}[1][Definition]{\begin{trivlist}
\item[\hskip \labelsep {\bfseries #1}]}{\end{trivlist}}
\newenvironment{example}[1][Example]{\begin{trivlist}
\item[\hskip \labelsep {\bfseries #1}]}{\end{trivlist}}

\newtheorem{theorem}{Theorem}[section]

\newtheorem{lemma}[theorem]{Lemma}

\newtheorem{proposition}[theorem]{Proposition}
\newtheorem{corollary}[theorem]{Corollary}

\newenvironment{thm_app}[1]{\noindent\textbf{Theorem~\ref{#1}.}}{\par\addvspace{\baselineskip}}
\newenvironment{lem_app}[1]{\noindent\textbf{Lemma~\ref{#1}.}}{\par\addvspace{\baselineskip}}

\newenvironment{prop_app}[1]{\noindent\textbf{Proposition~\ref{#1}.}}{\par\addvspace{\baselineskip}}

\DeclareMathOperator*{\argmax}{arg\,max}

\newcommand{\onebr}{\emph{One-Shot $\alpha$-BR }}

\title{Approximate Equilibrium and Incentivizing Social Coordination}
\author{Elliot Anshelevich \qquad Shreyas Sekar  \\ Rensselaer Polytechnic Institute, Troy, NY 12180\\\texttt{eanshel@cs.rpi.edu, sekars@rpi.edu.}}

\begin{document}


\maketitle
\begin{abstract}
We study techniques to incentivize self-interested agents to form socially desirable solutions in scenarios where they benefit from mutual coordination. Towards this end, we consider coordination games where agents have different intrinsic preferences but they stand to gain if others choose the same strategy as them. For non-trivial versions of our game, stable solutions like Nash Equilibrium may not exist, or may be socially inefficient even when they do exist. This motivates us to focus on designing efficient algorithms to compute (almost) stable solutions like Approximate Equilibrium that can be realized if agents are provided some additional incentives. Our results apply in many settings like adoption of new products, project selection, and group formation, where a central authority can direct agents towards a strategy but agents may defect if they have better alternatives. We show that for any given instance, we can either compute a high quality approximate equilibrium or a near-optimal solution that can be stabilized by providing small payments to some players. We then generalize our model to encompass situations where player relationships may exhibit complementarities and present an algorithm to compute an Approximate Equilibrium whose stability factor is linear in the degree of complementarity. Our results imply that a little influence is necessary in order to ensure that selfish players coordinate and form socially efficient solutions.

\end{abstract}

\section{Introduction}
Historically, the term coordination game has been applied to social interactions with positive network externalities. Typically, they are used to represent scenarios like the \textit{Battle of the Sexes} wherein self-interested agents benefit if and only if they choose the same strategy. Such a model, however, does not fully capture real-life situations like the adoption of technologies or opinions and the selection of activities where agents may eschew coordination if their personal preference for an alternative is very strong. For instance, a company may not adhere to common standards if the benefit from using their own proprietary technology far outweighs the gains from coordinating. Bearing this in mind, we consider the following broader interpretation of coordination games as `a class of games where agents' utilities increase when more people choose the same strategy as them'. Notice that this does not preclude agents from having intrinsic preferences for strategies.

Many social and economic interactions fall within our framework (see~\cite{jackson2002formation, galeotti2010network} for specific applications of coordination games and different types of interactions) and it is not surprising that the kind of games we are interested in have appeared in various guises throughout literature. Researchers have studied similar kinds of games in several settings including opinion formation~\cite{chierichetti2013discrete}, information sharing~\cite{kleinberg2013information}, coalition formation~\cite{feldman2012hedonic} and party affiliation~\cite{balcan2009improved,hoefer2007cost}, largely focusing on the existence and quality of stable solutions.
 

Given the significance of social coordination games, a natural question that arises is: Do instances of such games result in stable outcomes that are comparable to the social optimum, the solution maximizing social welfare. The \textit{somewhat} negative answer to this question that we provide serves as the starting point for our work as it highlights the need for incentivizing agents to form solutions that are beneficial for society. We attempt to answer the above question by articulating two fundamental drawbacks of coordination games, which naturally lead to the issue of influencing players to form ``good" solutions.
\begin{enumerate}
\item \textbf{Coordination Failures.} Coordination games suffer from Coordination Failures~\cite{cooper1999coordination} that result in agents becoming trapped in inefficient equilibria despite the existence of high-welfare equilibria. These situations arise when agents settle for less risky alternatives if they anticipate that other agents may not coordinate with them on what are potentially ``high-risk, high-reward" solutions. As an example, consider $N$ independent but complementary firms, each with a distinct preferred location, deciding on where to locate. Suppose that each company receives unit utility for choosing their favorite location and one more unit for every additional firm that choose to locate in the same area. Clearly an optimum and stable solution is one where all companies choose the same location. However, the outcome where each chooses their preferred location is also stable as no company could unilaterally deviate and profit. There is a large body of theoretical and experimental evidence (see~\cite{kosfeld2004economic} for a survey) that supports the hypothesis that agents may coordinate on inefficient outcomes even when better alternatives exist. 

\item \textbf{Non-existence of Equilibrium.} In instances where player relationships are asymmetric (they receive different gains from coordinating), a Pure Nash Equilibrium\footnote{We shall henceforth refer to Pure Nash Equilibrium as just Nash Equilibrium} may not even exist. The following example illustrates one such simple instance with three players and three strategies.
\begin{example}
\label{ex:1}
There are 3 players $i_1,i_2,i_3$. Player $i_j$ receives a utility of $\sqrt{2}$ if she chooses strategy $j$ and unit utility from strategy $j+1$ (addition here is modulo 3). Also, player $i_j$ receives coordination gains of 1 when player $i_{j+1}$ chooses the same strategy as her. Note that relationships are asymmetric so player $i_{j+1}$ receives no benefit for choosing the same strategy as player $i_j$.\\
It is not hard to verify that no Nash equilibrium exists for this instance. The reader is asked to refer to the proof of Proposition~\ref{prop:nonexistence} for details.
\end{example}
\end{enumerate}
It is evident that even in fairly simple coordination games, it may be necessary to guide agents to form desirable solutions. From a central point of view, a high social welfare is the most important requirement, but at the same time it is necessary that selfish agents do not deviate from these centrally promoted solutions. A key algorithmic challenge is therefore, computing stable outcomes with good social welfare that can be formed by providing each agent a small incentive. It is towards this end that we identify approximate equilibria as our primary solution concept.

\textbf{Approximate Equilibrium and Forming Stable Solutions.} An $\alpha$-Approximate Equilibrium is an outcome in which no player can improve their utility by a factor more than $\alpha$ by unilaterally deviating. Observe from the definition that if each player is provided additional benefits equaling a fraction $(\alpha-1)$ times their original utility, then no player would wish to deviate and the Approximate Equilibrium becomes a Nash Equilibrium. Alternatively, approximate stability corresponds to the addition of a switching cost that captures the inertia players may have in changing strategies unless the added benefit is large enough. In addition to being a simple generalization of Nash Equilibrium, approximate equilibria are also easily implementable or enforceable in natural settings as opposed to non-deterministic generalizations like Mixed Nash Equilibria. 

We focus on computing approximate equilibria with high social welfare establishing that although coordination games may not admit Nash Equilibrium, the addition of a relatively small amount of inertia to the game causes stable solutions with high social welfare to exist. We also consider group deviations via approximate strong equilibrium~\cite{feldman2009approximate}, computing solutions which are resilient to deviations by sets of players. A solution $\bm{s}$ is an $\alpha$-Approximate Strong Equilibrium if in any deviation by a set of players from $\bm{s}$, at least one of the players' utilities cannot improve by a factor more than $\alpha$ as compared to her original utility. We may also refer to $\alpha$ as the stability factor for all approximately stable solutions.

\subsubsection*{Formalizing our Model of Coordination Games}
We begin by considering a non-transferable utility game with $N$ players and $m$ distinct strategies. We assume that players have access to all strategies, therefore in any outcome of the game, a player's strategy $s_i \in \{1,\ldots, m\}$. Generalizing our examples from the previous section, we not only permit preferred strategies, but allow each player to have asymmetric preferences over the strategy set. Formally, player $i$ derives a utility of $w_i^k$ if she chooses strategy $1\leq k \leq m$. 

With regards to the coordination aspect, we now propose a framework where the benefits of coordination between two players do not depend on externalities. Specifically, suppose that $w(i,j)$\footnote{We take $w(i,j)$ and $w(j,i)$ to refer to the same entity} is the total coordination benefit when players $i$ and $j$ choose the same strategy and that this benefit is divided among the two players. Formally, player $i$ derives a utility of $\gamma_{ij}w(i,j)$ for coordinating with $j$ and player $j$ receives $\gamma_{ji}w(i,j)$. For a given instance of our game, the values $(\gamma_{ij}, \gamma_{ji}) \forall (i,j)$ are fixed and add up to one. So given a strategy vector $\bm{s}=(s_1,\cdots,s_n)$, the utility of a player $i$ has two components, $$u_i(\bm{s})=w_i^{s_i} + \sum_{s_j=s_i}\gamma_{ij}w(i,j).$$ We parameterize any instance on a factor $\gamma$ that captures the maximum asymmetry that exists in relationships. Formally $\gamma=\max{\frac{\gamma_{ij}}{\gamma_{ji}}}$ over all pairs $(i,j)$. We term this parameter, the \textit{Maximum Relationship Imbalance} (MRI). Since we are concerned with the quality of stable solutions, we define a social welfare function $u(\bm{s})$ to be the sum of player utilities. Mathematically, $u(\bm{s})= \sum_i u_i(\bm{s}) = \sum_i  w_i^{s_i} + \sum_{s_i=s_j}w(i,j).$ We define the optimum to be the solution maximizing social welfare and compare the quality of our solutions to the optimum welfare $OPT$. 


\subsection{Our Contributions.~}
In this work, we consider the following well-motivated question: can we implement solutions with high social welfare by providing each player some incentive to not deviate? Our main results answer this question in the affirmative, and more importantly we show that this is possible for every instance using one of our two incentivizing schemes. First, we present an algorithm based on greedy dynamics to compute a good quality, almost-stable solution. 
\begin{itemize}
\item(Theorem~\ref{theorem_approx}) There is a polynomial-time algorithm to compute an $\alpha$-Approximate Equilibrium ($\alpha \in [1.618,2]$) with a social welfare that is comparable to the optimum.
\end{itemize}
An approximate equilibrium corresponds to an easily realizable solution in the presence of either incentives, switching costs or players with inertia. Our second main result considers a complementary notion of stability: the minimum total payment to be provided to players so that they do not deviate from a desired high quality solution. For any given instance, if the algorithm of Theorem~\ref{theorem_approx} returns an $\alpha$-Approximate Equilibrium with social welfare $\rho_\alpha\cdot OPT$, then we show
\begin{itemize}
\item(Theorem~\ref{theorem_payment}) The optimum solution can be stabilized with a total payment of $\frac{\rho_\alpha}{\alpha-1}OPT$. 
\end{itemize}

Informally, this result tells us that if we can provide certain players supplementary utility, then with a finite budget we can stabilize the optimum solution. Given any instance, we first run the algorithm of Theorem~\ref{theorem_approx} to compute an $\alpha$-Approximate Equilibrium with social welfare $\rho_\alpha \cdot$ OPT. If $\rho_\alpha$ is large, then we have an almost-stable solution with high social welfare; else, if $\rho_\alpha$ is small, then with a budget of $\approx \rho_\alpha \cdot $OPT, we can force even the optimum solution to be stable. Together our two theorems imply something much stronger: we can always \emph{either} compute almost-stable solutions with high welfare, \emph{or} directly stabilize a good quality solution with small payments. 

In the process, we show some additional computability results for equilibria in our game.
\begin{itemize}
\item For $m=2$, we can compute a Nash Equilibrium and Strong Equilibrium in polynomial time.

\item For $m=3$, we can compute a $1.414$-Approximate Equilibrium and this factor is tight. i.e., for $\alpha < 1.414$, $\alpha$-Approximate Equilibrium may not exist.

\item For $m \geq 3$, a $1.618$-Approximate Equilibrium and a $2$-Approximate Strong Equilibrium always exist and we give algorithms to compute both.

\item We obtain tight lower bounds for the social welfare of the solution computed by the algorithm of Theorem~\ref{theorem_approx} in terms of $\gamma$ (Maximum Relationship Imbalance). When $\gamma$ is not too large, we show that this social welfare is comparable to the optimum. For instance, if relationships are not too asymmetric and $\gamma=2$, we can compute a $2$-Approximate Equilibrium whose social welfare is always at least forty percent of OPT and a $1.618$-Approximate Equilibrium whose social welfare is one-third of OPT. If $\gamma < 2$, then we can do much better. Table~\ref{tab:socwel} captures the social welfare of the solution returned by our algorithm for different values of $m$ and $\gamma$.
\end{itemize}

\textbf{Existence.} In general, Nash Equilibrium need not exist when $m \geq 3$. However, we identify sufficient conditions that encapsulate a broad class of Social Coordination Games (even when reward sharing is asymmetric) that guarantee the existence of a Nash Equilibrium. We do this by exhibiting a novel ordinal potential function for games where the benefits of coordination among players are closely correlated across relationships. This allows the mechanism designer to identify or create settings where repeated best-response by players always converges to a stable solution. We also show that Strong Equilibrium need not exist when $m \geq 3$ even under simplifying assumptions like $\gamma_{ij} = 1 \; \forall (i,j)$.

\begin{table}
\begin{tabular} {c|l l | l l}
 & \multicolumn{2}{c|}{$\alpha=2$} & \multicolumn{2}{c}{$\alpha=1.618$} \\
 $\gamma$ & $m=4$ & $m \rightarrow \infty$ & $m=4$ & $m \rightarrow \infty$ \\
 \hline
 1 & $0.57$OPT & $0.5$OPT & $0.424$OPT & $0.35$OPT \\
 2 & 0.47 & 0.4 & 0.37 & 0.29\\
 10 & 0.25 & 0.15 & 0.18 & 0.12    \\
\end{tabular}
\caption{Performance of our algorithm: The social welfare of our computed solution as a fraction of the optimum welfare for different values of $m$ (Number of strategies) and $\gamma$ (Maximum Relationship Imbalance).}
\label{tab:socwel}
\end{table}

\subsection{Related Work. }

Our model of social coordination is closely linked to two well known classes of games: non-transferable utility coalition formation and party affiliation. We begin by surveying the substantial literature in both these fields and examine the ties between our model and the games in these frameworks.

Hedonic games model players forming coalitions such that a player's utility depends only on the members of her own coalition. Our model can be embedded in this setting by considering a fixed number of non-anonymous coalitions and a set of players who are anchored, i.e., constrained to join only one particular group. Much of the work in hedonic coalition formation has focused on identifying conditions for the existence of stable solutions~\cite{banerjee2001core,bogomolnaia2002stability}. It is known, for instance, that if relationships are symmetric then the existence of Nash Equilibrium can be guaranteed by means of a potential function. Augustine et al.~\cite{augustine2011dynamics} consider a model similar to ours with the coordination benefits being submodular and characterize settings where Nash Equilibrium always exist. Although our games do not admit Nash Equilibrium, our results imply the existence of a stability concept that is slightly weaker than Nash stability for a large class of hedonic games with asymmetry. 

Another line of work has focused on quantifying the inefficiency of stable solutions~\cite{branzei2009coalitional} and on the computation of stable solutions~\cite{aziz2012computing,darmann2012group,gairing2010computing}. Although there have been a number of positive algorithmic results, the focus on approximating both stability and optimality has been limited. We also remark here that many of the defection models considered in the hedonic games literature are relevant for settings with an unbounded number of anonymous coalitions and may not be suitable for our class of games. With regards to influencing, the recent work on stabilizing desired coalitions via supplementary payments (the Cost of Stability)~\cite{bachrach2009cost} is similar to our direct payments technique, albeit in a transferable utility setting. 

Party affiliation games are a generalization of pure coordination games where players wish to coordinate with friends and {\textit{anti-coordinate} with enemies. On the other hand, in our model the friction is provided by the interplay between a player's individual preference and coordination. Party affiliation games with only two strategies and symmetric relationships have received considerable attention as Nash Equilibrium always exists in these settings although computing it is PLS-Complete~\cite{christodoulou2006convergence}. 
As a positive algorithmic result, Bhalgat et al.~(\cite{bhalgat2010approximating}) gave a polynomial time algorithm to compute a $(3+\epsilon)$-Approximate Equilibrium for such games. Most similar to our work in terms of motivation is the paper by Balcan et al.~(\cite{balcan2009improved}) who consider a model where only a limited number of players follow the centrally advertised strategy. Even though we look at only one aspect of party affiliation, our model is quite general as we do not impose any restriction on the number of strategies or player relationships. 	

\textbf{Other models of coordination in strategic settings.}
There has been a renewed interest in studying coordination games in networks from a theoretical perspective with an emphasis on identifying the kind of equilibrium outcomes that emerge. In particular, work has focused on characterizing the effect of several parameters on the equilibrium outcome including the cost of forming links~\cite{goyal2005network,jackson2002formation}, network structure~\cite{chwe2000communication}, level of interaction~\cite{morris2000contagion} and incomplete information~\cite{galeotti2010network}. The literature on coordination games is too vast and the reader is asked to refer to the paper by~\cite{galeotti2010network} for a more detailed survey.

\section{Preliminaries and Warm-up Results}
\label{sec:prelim}
In this section, we address some fundamental questions regarding stability and optimality in our Social Coordination Game (SCG) in order to gain a better understanding of our model. We then move on to the most simple version of our game where players only have two different strategies to choose from. For this special case, we present an algorithm that always returns a Nash Equilibrium whose social welfare is half of the optimum. We begin by casting our game in graph theoretic framework to compare our problem to existing optimization problems.  

%

\textbf{Social Coordination as a Network Game.} 
We can view our model as a game played on a complete graph $G=(V,E)$ where the nodes include the players and $m$ additional anchored nodes (which are constrained to choose only one strategy). Each directed edge $(i,j)$ has a weight $\gamma_{ij}w(i,j)$, the utility player $i$ derives from coordinating with $j$. Recall that player $j$ derives a utility of $\gamma_{ji}w(i,j)$ from coordinating with player $i$ and a utility of $w_j^k$ for choosing strategy $1 \leq k \leq m$.

We now define some additional notation to understand the social welfare of solutions better. We divide the social welfare of any strategy vector $\bm{s}$ into two components. $u(\bm{s}) = A(\bm{s}) + P(\bm{s})$, where $A(\bm{s})$ refers to the utility players receive due to their intrinsic preference for the chosen strategy and $P(\bm{s})$, the contribution of players' relationships to social welfare, both under $\bm{s}$. Mathematically $A(\bm{s}) = \sum_i w_i^{s_i}$ and $P(s) = \sum_{s_i = s_j}w(i,j)$. We denote by $best(i)$, the maximum utility that player $i$ derives from any one strategy and use $A_T$ to express $\sum_{i} best(i)$. Similarly, we use $P_T$ to refer to $\sum_{(i,j)} w(i,j)$, the maximum possible contribution of players' relationships to social welfare. Finally, we can upper bound $OPT$ by $A_T + P_T$.

It is not hard to see that the problem of maximizing social welfare is equivalent to the problem of dividing the graph into $m$ clusters to maximize the weight inside the clusters. This in turn is equivalent to the problem of minimizing the weight of the edges going across different clusters, which is the popular \textit{Multiway Cut} problem. We infer therefore, that the problem of maximizing social welfare in any given instance of the Social Coordination Game is NP-Hard.


\begin{proposition}
For $m > 2$, computing the optimum solution for an instance of the Social Coordination Game is NP-Hard even for a fixed value of $\gamma$.
\end{proposition}
The optimization version of our problem was considered for undirected graphs in~\cite{langberg2006approximation} as the \textit{Multiway Uncut problem}. They exhibited a 0.8535 approximation algorithm for the same using an LP-rounding based approach. We extend their results to the directed version via a simple reduction. 

\begin{proposition}
There exists a polynomial-time algorithm to compute a solution of the SCG such that its social welfare is at least a fraction 0.8535 of $OPT$.
\end{proposition}
\begin{proof} Given an instance of the SCG, we create an undirected graph with weight $w(i,j)$ on edge $(i,j)$ and run the algorithm of Langberg et al.\end{proof}
\subsubsection*{Social Coordination Games with Two strategies}
The simplest form of Social Coordination is one where players get to choose between two distinct strategies, which has received considerable attention in literature. For this special case, we are able to guarantee the existence of a Nash Equilibrium by constructing one for any given instance. Our actual result is much stronger, we give a sequence of best-response moves from any given starting state that always results in a Nash Equilibrium.  By Proposition~\ref{prelim_poa} (shown in the next section), we are guaranteed that this solution is at least half as good as $OPT$ in terms of social welfare. 

\begin{algorithm}[htbp]
\caption{Computing Nash Equilibrium when $m=2$}
\label{alg_m2ne}
\algsetup{indent=2em}
\begin{algorithmic}[1]
\STATE Allow players to deviate from strategy $1$ to $2$ if it is their best response until no player wants to deviate from strategy $1$.
\STATE Repeat this for the other strategy until no player can improve her utility by deviating from strategy $2$.
\end{algorithmic}
\end{algorithm}
\begin{proposition}
The above algorithm returns a Nash Equilibrium from any given starting state $\bm{s_0}$ such that its social welfare is at least half of $OPT$.
\end{proposition}
\begin{proof} At the end of Step 1, no player would want to deviate from strategy $1$ to $2$ by definition. These players would not want to deviate after step 2 either as their utility only increases. By definition, players cannot improve their utility by deviating from strategy 2 and thereby we have a Nash Equilibrium.\end{proof}

A similar algorithm was used in~\cite{chierichetti2013discrete} to compute Nash equilibria in a cost minimization game with two preferences.

\section{Existence and Quality of Stable Solutions}
In the previous section, we showed that if every player has exactly two strategies to choose from, we can always compute an equilibrium in polynomial time. Unfortunately, we are not so lucky even if players have one additional strategy to choose from as Nash Equilibrium may not exist. However, in this section, we show sufficient conditions to guarantee the existence of equilibrium when $m \geq 3$. In the context of our central theme of incentivizing coordination, this is useful in settings where the central designer may not be able to provide any direct incentives but can exert some control over the parameters of the game so that natural game play (best-response dynamics) results in stable solutions. 

Although similar games like Party affiliation and Discrete opinion formation also admit potential functions~\cite{bhalgat2010approximating, chierichetti2013discrete}, our result does not follow from theirs as our potential function is not an exact potential function. Further, our conditions also capture scenarios where the benefits of coordination are not equally shared. First, we formalize the example given in the introduction regarding the non-existence of equilibrium.
\begin{proposition}
\label{prop:nonexistence}
There exist instances of the Social Coordination Game with three strategies where Nash Equilibria do not exist.
\end{proposition}
\begin{proof}
\begin{figure}[!htbp]
\centering
\label{figure:asymmetricnoeq}
\includegraphics[width=.4\linewidth]{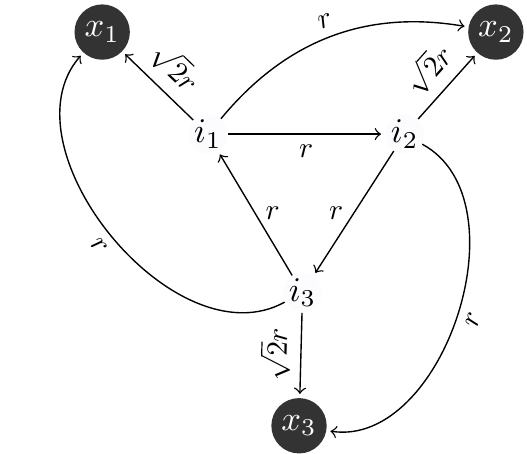}
\caption{Instance of the SCG with 3 strategies (represented by shaded anchor nodes) and three players ($i_1, i_2, i_3$). Player $i_j$ receives a utility of $\sqrt{2}r$ if she chooses strategy $x_j$ and $r$ utility from strategy $x_{j+1}$. Also, player $i_j$ receives coordination gains of $r$ when player $i_{j+1}$ chooses the same strategy as her.}
\end{figure}
In any equilibrium, each player receives a utility of at least $\sqrt{2}r$ or else they could deviate to their favorite strategy. It is not hard to see that in any strategy other than the one where all players are separated, at least one player's utility is strictly less than $\sqrt{2}r$ and hence she would deviate. Therefore, in every strategy at least one player can deviate and improve her utility by a factor of $\sqrt{2}$.
\end{proof}

Although existence cannot be guaranteed for the general case, the following natural assumption captures a broad class of SCGs for which Nash Equilibrium always exists.

\textbf{(Correlated Coordination Condition)} A given instance of the SCG is said to satisfy this condition if $\exists$ a vector $\bm{\gamma}=(\gamma_1, \cdots, \gamma_N)$ such that $\forall (i,j)$ with $w(i,j) >0$, we have $\gamma_{ij}=\frac{\gamma_i}{\gamma_i + \gamma_j}$ and $\gamma_{ji}=\frac{\gamma_j}{\gamma_i + \gamma_j}$.

Informally, each player has an intrinsic weight or an influence factor $\gamma_i$. When two players (say $i$ and $j$) coordinate, the benefit that any one player derives depends on her relative weight with respect to the other player. For any two players $i$ and $j$, if $\gamma_i \geq \gamma_j$, then player $i$ receives a greater benefit due to coordinating with player $j$. Notice that this definition does not impose any restriction on how asymmetric relationships can be. We now claim that games which obey this condition not only have Nash Equilibrium but in fact admit an ordinal (inexact) potential function~\cite{monderer1996potential} that ensures that best-response always converges to an equilibrium. 

\begin{theorem}
\label{theorem_ccexistence}
Social Coordination Games which obey the Correlated Coordination (CC) Condition admit a potential function. Therefore, best-response dynamics always converge to a Nash Equilibrium in such games. 
\end{theorem}
\begin{proof}
Given a social coordination game with weights $w(i,j), \forall (i,j)$ and $\bm{\gamma}=(\gamma_1, \cdots, \gamma_N)$, that obeys the CC condition, we claim that the following is an inexact potential function for this game. 
$$\Phi(\bm{s}) = \sum_{i}\frac{w_i^{s_i}}{\gamma_i} + \sum_{s_i=s_j}\frac{w(i,j)}{\gamma_i + \gamma_j}.$$

To prove this, we need to show that whenever a player makes an improving move, the potential function increases and vice-versa. Consider a strategy vector $\bm{s}$ and a player $i$ whose utility increases, when she deviates from $s_i$ to $s'_i$. The change in her utility is $u_i(s'_i, s_{-i}) - u_i(s_i, s_{-i}) = (w_i^{{s'_i}} + \sum_{s_j = s'_i} \gamma_{ij}w(i,j)) - (w_i^{{s_i}} + \sum_{s_j = s_i} \gamma_{ij}w(i,j))$, which is positive. As our games obey the CC condition, we can substitute $\gamma_{ij}$ with $\frac{\gamma_i}{\gamma_i+\gamma_j}$ and obtain,
$$ w_i^{{s'_i}} + \sum_{s_j = s'_i} \gamma_{i}\frac{w(i,j))}{\gamma_i + \gamma_j} - w_i^{{s_i}} - \sum_{s_j = s_i}  \gamma_i \frac{w(i,j)}{\gamma_i + \gamma_j} > 0.$$ This implies that, 
$$\displaystyle \gamma_i (\frac{w_i^{{s'_i}}}{\gamma_i} + \sum_{s_j = s'_i} \frac{w(i,j)}{\gamma_i + \gamma_j} - \frac{w_i^{{s_i}}}{\gamma_i} - \sum_{s_j = s_i}  \frac{w(i,j)}{\gamma_i + \gamma_j} ) > 0.$$ Since $\gamma_i$ is always positive, we conclude that the term inside the parenthesis must be positive as well. Therefore, the potential function strictly increases when players make a better-response move. The other direction is similar, if $(\frac{w_i^{{s'_i}}}{\gamma_i} + \sum_{s_j = s'_i} \frac{w(i,j)}{\gamma_i + \gamma_j} - \frac{w_i^{{s_i}}}{\gamma_i} - \sum_{s_j = s_i}  \frac{w(i,j)}{\gamma_i + \gamma_j}) > 0$, then we multiply both sides by $\gamma_i$ yielding the change in player $i$'s utility on the left hand side. This indicates that any deviation by a player which increases the potential function must be a strictly improving move. We conclude that these games admit a potential function. Therefore, best-response dynamics must always converge to a Nash Equilibrium. \end{proof}

Among other scenarios of interest, the CC conditions also captures games where player relationships are symmetric but players may derive any arbitrary utility from choosing a strategy($\gamma_i = \gamma_j \; \forall (i,j)$). Therefore, if the benefits of a relationship are split equally among the players, stable solutions always exist.

\subsubsection*{Quality of Stable Solutions}
Although Nash Equilibrium exists for special cases, it is not clear whether these stable solutions would be centrally desirable in terms of social welfare. We now formalize the notion of ``Coordination Failures" that we mentioned earlier and show that in general Coordination Games, both the best and worst Nash Equilibrium can be a factor $m$ away from the social optimum. In other words, even when Nash Equilibrium exists, all stable solutions may have a social welfare that is only a small fraction of $OPT$. This reinforces the need for some external influence in order to ensure solutions that are closer to $OPT$ in terms of social welfare. 

\begin{proposition}
\label{prelim_poa}
The Price of Anarchy (PoA) for Social Coordination Games over all instances is at most $m$.
\end{proposition}
The Price of Anarchy is the ratio of the social welfare of the optimum to that of the worst Nash Equilibrium (equilibrium with the lowest social welfare for a given instance) over all instances.

\emph{(Proof Sketch)} We use a refinement of the popular smoothness technique~\cite{roughgarden2009intrinsic} known as semi-smoothness (see~\cite{lucier2011gsp, anshelevich2013assignment} for more details) for proving the above bound. The full proof is present in the appendix. We show that our Social Coordination Game is $(\frac{1}{m},0)$-semi-smooth for all instances. This immediately gives us an upper bound of $m$ for the Price of Anarchy. $\blacksquare$

Note that by the generality of the smoothness technique, we are guaranteed that the Price of Anarchy is $m$ for solution concepts that are much more general than Nash Equilibrium including Mixed Nash Equilibrium and Coarse Correlated Equilibrium. We now show that the above factor is tight by giving a matching lower bound for the best Nash Equilibrium. 

\begin{proposition}
\label{prelim_pos}
The Price of Stability (PoS) for the SCG is at least $m$ and this is tight.
\end{proposition}
The Price of Stability for a given game is the maximum ratio of the social welfare of the optimum and the best Nash Equilibrium over all instances. 

\begin{proof}
\textbf{(Example)} Consider an instance with $m$ strategies and $m$ players, $i_1 \cdots i_m$. Player $i_1$ derives a utility $r$ from strategy $i$ and an additional utility of $r$ for every other player who chooses the same strategy as her. Players $i_j$ ($j>1$) receive a utility of $2\epsilon$ from choosing strategy $j$ and an additional utility $\epsilon$ if player $i_1$ chooses the same strategy. In the social optimum, all players choose strategy $1$, and we have $u(OPT) =  mr + (m-1)\epsilon$. This instance has $m$ different Nash Equilibria based on which strategy player $i_1$ chooses and every Nash Equilibrium has a social welfare of $r + 2(m - 1)\epsilon$ . In the limiting case as $\epsilon \rightarrow 0$, the ratio between the social welfares approaches $m$.\end{proof}

This implies that in instances where Nash Equilibrium exist, all stable solutions can have a social welfare that is much worse than the optimum (if $m$ is large). Propositions~\ref{prelim_poa}, \ref{prelim_pos} also tell us that if we are able to limit the number of choices available to self-interested agents, then all stable solutions exhibit high social welfare.

\subsubsection*{Quality of Stable Solutions for the Symmetric Coordination Game}
We return to the special case when player relationships are symmetric (rewards are split-up equally), which is of considerable interest as it occurs in several scenarios. Notice that this case can be expressed as a SCG with $\gamma_{ij}=1 \; \forall(i,j)$ and therefore, every player receives a utility of $\frac{w(i,j)}{2}$ from each relationship. We show that symmetry improves the situation by giving a tight upper bound of $(2-\frac{1}{m})$ on the Price of Stability. Unfortunately, this best equilibrium may not be computable. In the next section, we present an algorithm to compute an approximate equilibrium that is almost as good as the best Nash Equilibrium. First observe that the following is an exact potential function for the SCG with symmetric relationships.
$$\Phi(\bm{s}) = \sum_i w_i^{s_i} + \frac{1}{2}\sum_{s_i=s_j}w(i,j).$$

\begin{proposition}
\label{pos_symmetric}
The price of stability is at most $(2-\frac{1}{m})$ in social coordination games with symmetric relationships, i.e. $\gamma=1$. Further, this bound is tight.
\end{proposition}
\emph{(Proof Sketch)} The complete proof can be found in the Appendix. We prove the first part of the proposition by proving a much stronger statement on best-response dynamics: that the loss in social welfare from certain `good starting states' after any best-response sequence is upper bounded by the factor $2 - \frac{1}{m}$. Then all we need to show is that $OPT$ belongs to these good states, by proving that $A(OPT) \geq \frac{A_T}{m}$. We also show a tight example where the PoS is exactly $(2-\frac{1}{m})$, which completes the proof. $\blacksquare$

The proposition also implies that in the special case with two strategies, the polynomial-time best-response sequence described previously with OPT as a starting state (can be computed using any algorithm for Min-Cut) leads to a Nash equilibrium whose social welfare is at least two-thirds of OPT. 

\begin{corollary}
\textbf{Algorithm 1} for two strategy SCGs with symmetric relationships returns a Nash Equilibrium that is at least two-thirds of the optimum when $OPT$ is used as the starting state.
\end{corollary}

\section{Approximate Equilibrium}
We now come to our main results. In the previous sections, we showed that Nash Equilibrium may not exist for Social Coordination Games and even when it does, its quality may not be comparable to the optimum. This motivates us to relax our notion of stability and consider approximately stable solutions in order to obtain better guarantees on the social welfare. Our first main result in this section is a polynomial-time algorithm to compute a $1.414$-Approximate Equilibrium for SCGs with three strategies. When $m>3$, we give an easily implementable algorithm to compute an $\alpha$-Approximate Equilibrium for $\alpha \in [\phi, 2]$ ($\phi \approx 1.618$). As we showed that $\alpha$-Approximate Equilibria for $\alpha < 1.414$ do not exist for general SCGs, our stability results are nearly tight (up to a factor of $\approx 1.14$). Recall that in an $\alpha$-Approximate Equilibrium, no player can deviate from one strategy to another and increase her utility by more than $\alpha$ times her previous utility.

\textbf{Three Strategies:} For $m=3$, recall that in the example of Figure~\ref{figure:asymmetricnoeq}, in every state there is at least one node that can deviate and improve its utility by a factor of $\sqrt{2}$. We show now that this bound of $\sqrt{2}$ is tight by exhibiting an algorithm for three groups that always returns a $\sqrt{2}$-Approximate Nash equilibrium. 
\begin{proposition}
\label{prop:1414approx}
The following algorithm returns a $\sqrt{2}$-Approximate Equilibrium for all instances with three strategies and no algorithm can give a better approximation factor over all instances as this ratio of $1.414$ is tight.
\end{proposition}
\begin{enumerate}
\item Run \emph{Algorithm 1} for $m=2$ (for say strategies $1$ and $2$) and all players so that no player wants to deviate from one of those two strategies to the other. 
\item While there exists a player, who can deviate from either strategy $1$ or strategy $2$ and improve her utility by a factor $\sqrt{2}$ or more, allow her to deviate to strategy $3$. Once that player deviates, run the algorithm for $m=2$ (for strategies $1$ and $2$) so that players choosing these two strategies are stable with respect to each other.
\item Repeat step 2 until no player can to deviate to 3 and increase her utility by a factor $\sqrt{2}$.
\end{enumerate}
In other words, at every stage of the algorithm, players in strategies $1$ and $2$ are stable with respect to each other. Players deviate to $3$ only when it is their best-response and it improves their utility by the given factor. The algorithm terminates in polynomial-time as in each round the size of group 3 increases by one. We prove that its stability factor is $\sqrt{2}$ in the appendix.

Moving on to the general case, we give an algorithm to compute a $\alpha$-Approximate Equilibrium and characterize the trade-off between stability and optimality by showing how the social welfare increases when we decrease the stability factor from $\alpha=2$ to $\alpha=\phi$. Figure~\ref{figure:tradeoff_symmetric} illustrates this near-linear trade-off between social welfare and $\alpha$. As mentioned, our bounds for social welfare depend on $\gamma$, the Maximum Relationship Imbalance (MRI), that measures the maximum asymmetry over all relationships. We explicitly use this factor $\gamma$ in our welfare as often real-life relationships are bounded in their asymmetry; if one player receives large gains from coordinating with another, then the second player derives at least some fraction of that benefit. Recall that $\gamma=1$ denotes the case with symmetric relationships, as $\gamma$ increases, reward sharing becomes more and more asymmetric. 

We begin by presenting a simple algorithm called \emph{one-shot $\alpha$-BR}, which allows each player to play her best-response at most once as long as the improvement in utility is by a multiplicative factor $\alpha \geq 1$. This natural algorithm has desirable properties in terms of both stability and welfare and acts as a subroutine for many of our results in this section and the next. First, we show the following lemma.

\begin{algorithm}[htbp]
\caption{One-shot $\alpha$-BR}
\label{alg_1shotbr}
\algsetup{indent=2em}
\begin{algorithmic}[1]
\REQUIRE A starting strategy $k_0$.
\STATE Set all players' initial strategy to be $k_0$.
\STATE While there exists a player whose current strategy is $k_0$, who can improve her utility by at least a factor $\alpha$ by deviating to another strategy, allow her to perform best-response. \label{alg1shotbr_loop}
\STATE Return the final solution $\bm{s}$.
\end{algorithmic}
\end{algorithm}

\begin{lemma}
\label{1shotbr_stabilityfactor}
The \emph{One-shot $\alpha$-BR} algorithm returns a $\left(\max\left(\alpha, \frac{1}{\alpha} + 1\right)\right)$-approximate equilibrium, for any starting strategy $k_0$.
\end{lemma}
That is, for $\alpha \in [\phi,2]$, the \emph{One-Shot $\alpha$-BR} algorithm returns an $\alpha$-Approximate Equilibrium.
\begin{proof}
We introduce some notation here in order to analyze our \emph{one-shot $\alpha$-BR algorithm}. First, assume that the algorithm proceeds in steps, such that at each time step only one player deviates from $k_0$. Clearly, the algorithm terminates after at most $N$ time steps. Let $X^t_k$ denote the set of players whose strategy is $k$ at time step $t$ and $s_{-i}^t$, the strategy vectors of all players other than $i$ at that time. Let $X_k(\bm{s})$ refer to the players choosing strategy $k$ in the final solution. Our first observation is that the number of players choosing a particular strategy other than $k_0$ cannot decrease with time. 

Now, consider any player $i$ such that $i$ deviated to some strategy $k = s_i$ at time $t$. Then, we have
\begin{equation}
\label{eqn_appinvariant}
  \begin{aligned}
  u_{i}(s_i, s_{-i}^t) & \geq & u_i(k', s_{-i}^t)  & \quad \forall k' \neq s_i\\
  u_{i}(s_i, s_{-i}^t) & \geq & \alpha u_i(k_0, s_{-i}^t)
  \end{aligned}
\end{equation}

Since $i$ performs her best-response, she derives more utility from her strategy than she would by deviating to any other strategy. Each player who deviates improves her utility by a factor $\alpha$ by definition, which gives us the second inequality above. Now suppose that $i$'s best-response under $\bm{s}$ is to move to some strategy $k'$. Observe that $X_{k'}(\bm{s}) \subseteq X^t_{k'} \cup X^t_{k_0}$ as the new players choosing $k'$ after time $t$ can only come from the $k_0$. The utility that $i$ gets from this best-response strategy $k'$ can be upper bounded as $u_{i}(k',\bm{s_{-i}}) \leq u_i(k', s_{-i}^t) + u_i(k_0, s_{-i}^t) \leq  u_{i}(s_i, s_{-i}^t) + \frac{ u_{i}(s_i, s_{-i}^t)}{\alpha} \leq (1+\frac{1}{\alpha})u_i(\bm{s})$. The last two inequalities come from invariant~\ref{eqn_appinvariant} and the fact that $X^t_k \subseteq X_k(\bm{s})$ respectively, giving us the stability factor for players who deviated from $k_0$.

By definition, the algorithm terminates when no player whose strategy is $k_0$ can deviate and improve her utility by a factor $\alpha$. Therefore, for any player $i \in X_{k_0}(\bm{s})$ and other strategy $k' \neq k_0$, we have $\alpha u_i(\bm{s}) \geq  u_{i}(k',s_{-i})$. This means that any arbitrary player in $N$ can perform her best-response and improve her utility by a factor no greater than $\max(\alpha, \frac{1}{\alpha} + 1)$, which completes the proof. 
\end{proof}

The following lemmas give us an idea about the social welfare of the solution returned by the \emph{One-Shot $\alpha$-BR} algorithm. In particular, we see that players get at least as much utility as they would by choosing their favorite strategy (in the absence of coordination). Further, the final utility is quite close to the utility of the starting state where all players choose the same strategy $k_0$. Recall that $A_T$ is the total utility received by players from their preferred strategy, and $A(\bm{s})$ and $P(\bm{s})$ are the utilities due to intrinsic preference and coordination in any given strategy $\bm{s}$.

\begin{lemma}
\label{1shotbr_swanchor}
The \emph{1-shot $\alpha$-BR} algorithm, beginning with any starting strategy $k_0$, returns a solution whose social welfare is at least $\displaystyle\frac{A_T}{\alpha}$
\end{lemma}

\begin{lemma}
\label{lemma_asymmetricsw}
The \emph{one-shot $\alpha$-BR} algorithm with $\bm{s_0}$ as the starting state (where all players choose a strategy $k_0$), results in a social welfare of at least $A(\bm{s_0}) + \frac{\alpha}{\gamma+1}P(s_0)$, where $\gamma$ is the Maximum Relationship imbalance.
\end{lemma}
\emph{(Proof Sketch)} The full proofs can be found in the appendix. Notice that the algorithm allows each player to deviate exactly once, so each player gets as much as she would (discounted by $\alpha$) by choosing her preferred strategy. Summing up over all players gives us Lemma~\ref{1shotbr_swanchor}. For the second lemma, observe that if a player $i$ deviates, and the utility of other players drops by some amount $y$, then player $i$'s original utility was at least $\frac{1}{\gamma} y$. Since she performs a best response move, her new utility is at least the same amount. Therefore, for every $(\gamma+1)$ units of social welfare lost due to coordination, at least one unit of social welfare is gained, giving us the desired bound. $\blacksquare$

We are now in a position to show our main result. Since the above algorithm returns a solution whose social welfare is close to that of the starting state, it seems to natural to select a starting state with a high social welfare. We choose the strategy $k^*$ that players have the maximum intrinsic preference for, i.e., $k^* = \argmax_{1 \leq k \leq m}(\sum_i w_i^{k})$. 

\begin{theorem}
\label{theorem_approx}
The following algorithm returns an $\alpha$-Approximate Equilibrium for $\alpha \in [\phi, 2]$ whose social welfare is approximately at least a fraction $\displaystyle \approx \max(\frac{\alpha}{\gamma+3}, \frac{1}{m})$ of the optimum social welfare. \end{theorem}
\begin{quote}
\textbf{Algorithm:} ``For a given $\alpha$, run \emph{One-shot $\alpha$-BR} and \emph{One-shot $\frac{1}{\alpha-1}$-BR} with $k^*$ as the starting strategy. Let the returned solutions be $\bm{s_1}$ and $\bm{s_2}$. Return the solution among these two with greater social welfare."
\end{quote}
\textbf{Discussion.} The exact social welfare is: $\displaystyle\frac{\alpha - 1}{1 + \frac{(\gamma+1)}{\alpha}(\alpha - (1 + \frac{1}{m}))}$ times $OPT$ when $(\gamma+1) \leq \alpha m$, and $\max(\frac{\alpha}{\gamma+1}, \frac{\alpha-1}{(m-1) + (\alpha-1)})$ times OPT otherwise. We attempt to break-down the dependence of social welfare on different parameters. First, notice from the approximate formula $\frac{\alpha}{\gamma+1}$ that the social welfare increases almost linearly with $\alpha$. In other words, as is usually the case, sacrificing a little individual stability results in greater overall well being. At the same time, we see that social welfare decreases when $\gamma$ becomes larger, i.e., when relationships become 
more asymmetric. Therefore, for a designer, there are two possible measures to ensure socially efficient outcomes: \textbf{(i)} imposing a higher switching cost on agents, \textbf{(ii)} splitting the rewards of coordination almost equally. 

Our algorithm provides good guarantees when $\gamma$ is not too large. For instance, when the benefits from coordination are off by at most a factor of 2 ($\gamma \leq 2$), our algorithm returns a $1.618$-Approximate Equilibrium that is almost one third of $OPT$ even if the number of available alternatives is arbitrarily large. Contrast this with the result in Proposition~\ref{prelim_poa} that states as $m \to \infty$, the social welfare of equilibria becomes infinitely worse off than the optimum. For the same $\gamma$, if we sacrifice some stability and move to a $2$-Approximate Equilibrium, we get more than forty percent of the optimum in terms of welfare. Finally, the social welfare is bounded by $\approx \frac{1}{m}$, which indicates that even when $\gamma$ is large, our solutions are at least as good as that of the Nash Equilibrium (when they exist). In other words, running our algorithm is always preferable to letting agents form their own strategies. Figure~\ref{figure:tradeoff_asymmetric} illustrates how the social welfare drops when relationships become more asymmetric.

\emph{(Proof Sketch)} Both the procedures in the algorithm result in an $\alpha$-Approximate Equilibrium. Applying Lemma~\ref{1shotbr_stabilityfactor} to the solution returned by the \emph{1-shot $\frac{1}{\alpha-1}$-BR} procedure, we get that its stability factor is $\max(\frac{1}{\alpha-1}, \frac{1}{\frac{1}{\alpha-1} }+1)$, that is $\max(\frac{1}{\alpha-1}, \alpha)$. $\bm{s_1}$ is, therefore, an $\alpha$-approximate Nash equilibrium as $\alpha \geq \frac{1}{\alpha-1}$ for $\alpha \in [\phi,2]$. Similarly, applying the lemma to $\bm{s_2}$, we get that its stability factor is $\max(\alpha, \frac{1}{\alpha}+1)$, which is also $\alpha$ for the given range. This establishes the stability factor. 

The social welfare of the state where all players have chosen strategy $k^*$ is at least $\frac{A_T}{m} + P_T$ as players have greater overall intrinsic preference for this strategy than any other strategy. Applying Lemma~\ref{lemma_asymmetricsw} to our algorithm's $\alpha$-BR procedure, we get that the final social welfare of the solution, $u(\bm{s})$, is at least $\frac{A_T}{m} + \frac{\alpha}{\gamma+1}P_T$. Applying Lemma~\ref{1shotbr_swanchor} to our \emph{One-shot $\frac{1}{\alpha-1}$-BR} algorithm, we get that our computed solution has a social welfare of at least $(\alpha-1)A_T$. Taking the best solution among the two and finding the point where it is minimum gives us the desired bound. $\blacksquare$.

We return to the case where $\gamma=1$. We previously proved that the best Nash Equilibrium has a social welfare which is a factor $2 - \frac{1}{m}$ less than OPT. Substituting, $\gamma=1$ in the above equation, we see our $2$-Approximate Equilibrium has the same social welfare as the best Nash Equilibrium.

\begin{corollary}
For the special case when coordination benefits are split equally ($\gamma=1$), our algorithm returns a $2$-Approximate Equilibrium that is more half as good as the optimum and a $1.618$-Approximate Equilibrium that is at least a fraction $0.35$ of OPT.
\end{corollary}

\begin{figure}
\hfill
\subfigure[Symmetric Group Formation($\gamma=1$ case): Trade-off between welfare and stability.]{\label{figure:tradeoff_symmetric} \includegraphics[width=.45\linewidth]{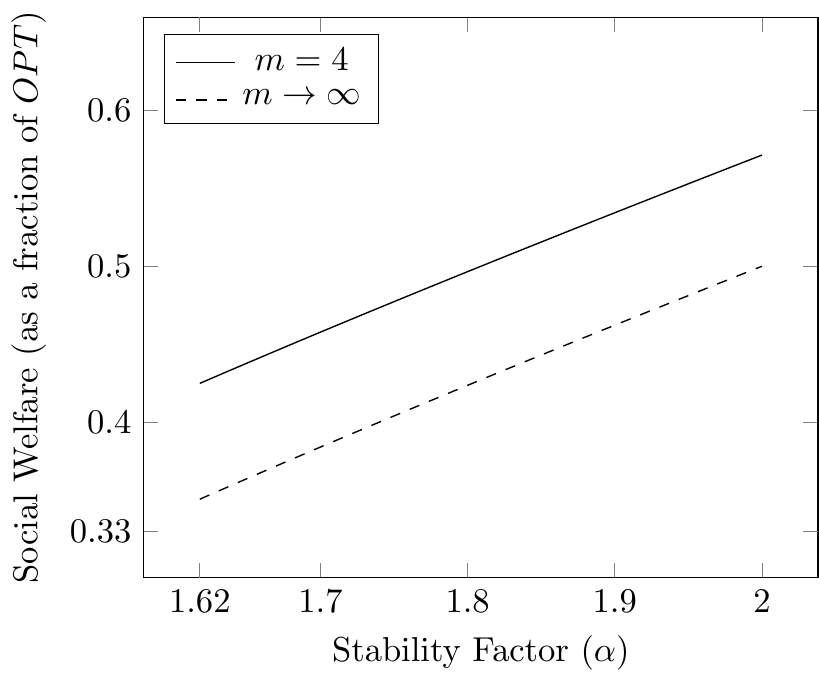}}
\hfill
\subfigure[Asymmetric Group Formation Game: How the asymmetry of the graph affects welfare when $\alpha=1.618$.]{\label{figure:tradeoff_asymmetric} \includegraphics[width=.45\linewidth]{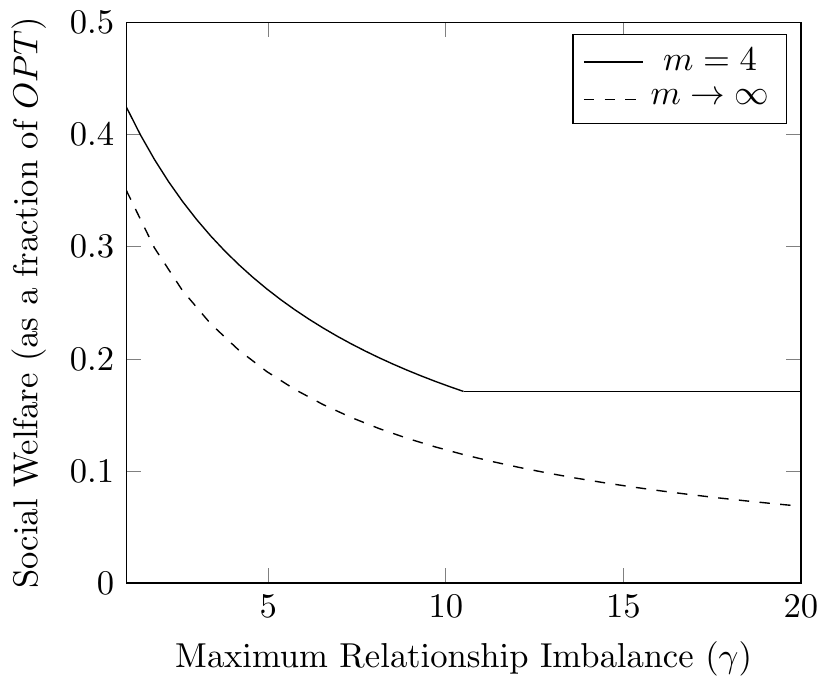}}
\hfill
\caption{Social Welfare as a function of $\alpha$ and $\gamma$.}
\end{figure}

\subsection{Strong Equilibrium}
In any game involving coordination, it becomes important to consider Strong Nash stability as coalitions of players may be able to cooperate when switching to different strategies. Given our focus on computing good, approximately stable solutions, it becomes imperative to design solutions where for any group of players who can deviate and improve their utility, at least one player's utility increases by a factor non more tha $\alpha$. We term this $\alpha$-approximate Strong equilibrium. In this section, we show both (non-)existence and computability results for Strong and approximately Strong equilibria.

We begin by looking at Social Coordination games with only two strategies. Our first result is that any SCG with two strategies admits a Strong Nash Equilibrium. We give a constructive proof for this by extending Algorithm 1 from Section~\ref{sec:prelim} to also compute SNE.

\begin{proposition}
\label{claim_se2groups}
Any instance with $m=2$ strategies admits a Strong Nash equilibrium. The following algorithm computes a Strong Nash Equilibrium.
\end{proposition}
\begin{quote}
\begin{enumerate}
\item Let $\bm{s_0}$ be the starting state with all players present in $1$.
\item While there exist any group of players who can deviate from strategy $1$ to strategy $2$ such that each player's utility increases strictly from their previous utility, allow this group to deviate.
\end{enumerate}
\end{quote}
\begin{proof} The proof is similar to that of Algorithm 1. Observe that by definition, no set of players can deviate from 1 and improve their utility. If a set of players from 2 deviate to 1, then it is not hard to see that the player in this set who deviated at the earliest time step in the algorithm cannot be improving her utility.\end{proof}

For three or more strategies, strong equilibrium may not exist even for the simplest special case of Symmetric SCGs where $\gamma=1$ and Nash Equilibrium is guaranteed to exist. The non-existence of strong equilibrium motivates us to look at approximate Strong equilibria and how to compute them efficiently. Our main result in this section is the efficient computation of a $2$-approximate Strong equilibrium with good social welfare. First, we exhibit a non-trivial example to prove the non-existence of Strong Nash Equilibrium.


\begin{proposition}
\label{claim_seexistence}
There are instances of the symmetric Social Coordination Game with $m=3$ groups where Strong equilibrium does not exist. However, whenever Strong equilibrium exists, its social welfare is no less than half of OPT when player relationships are symmetric.
\end{proposition}
\begin{proof}
\begin{figure}
\centering
\includegraphics[width=.8\linewidth]{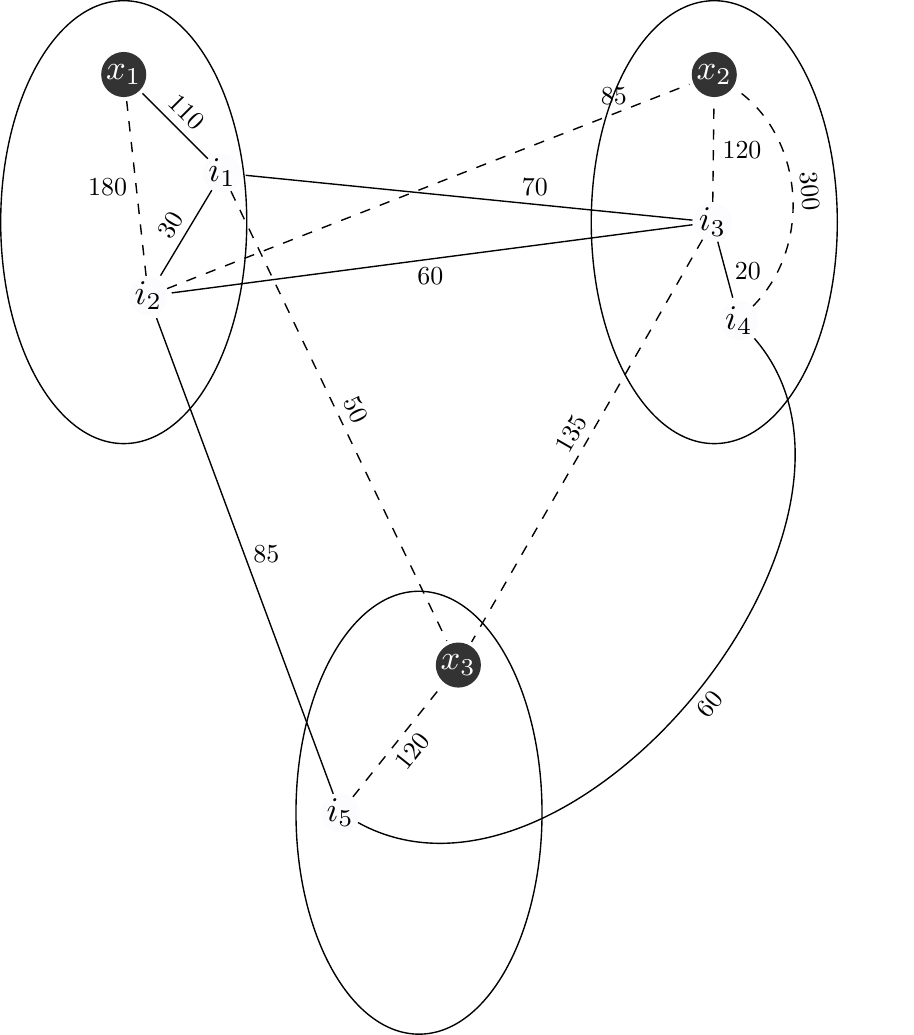}
\caption{Example: Symmetric Social Coordination Game with three strategies marked by ellipses (anchored players are shaded) and five players. For the purpose of clarity, edges between players are solid, whereas the edges between players and anchors are dashed. In the current state, all players within an ellipse have their strategies set to be that corresponding strategy.} \label{figure:symmetricnosne}
\end{figure}
Consider the instance shown in Figure~\ref{figure:symmetricnosne} with $3$ strategies and $5$ players. The figure makes use of the graph theoretic framework defined in~\ref{sec:prelim}. As the instance is symmetric, we assume that utility recieved due to coordination by both the players equals the weight on the edge ($\frac{w(i,j)}{2}$). In any Nash equilibrium for this instance, two players from two different strategies can always deviate to the third strategy and improve their utilities. For instance, consider the state shown in the figure, $i_2$ and $i_5$ can deviate to strategy $2$ and both increase their utilities. In any state where $i_2$ is not present in strategy $1$, $i_1$ and $i_3$ can jointly deviate to 3. In all other states, at least one player can perform best-response by herself. We prove the second part of the claim in the appendix. \end{proof}

\begin{proposition}
\label{prop_se2app}
For $\alpha \geq 1$, the \emph{One-shot $\alpha$-BR} returns a $(1+\alpha)$-Approximate Strong Equilibrium.
\end{proposition}
We remark here that for $\alpha=1$, we get a $2$-Approximate Strong Equilibrium whose social welfare is at least half of that mentioned in Theorem~\ref{theorem_approx}. Therefore, our \emph{One-shot $\alpha$-BR} algorithm also offers resilience against group deviations. The full proof of the above proposition is present in the Appendix.

\emph{(Proof Sketch)} Let $\bm{s}$ be the solution returned. For any player $i$ whose strategy is $k_0$, for any other strategy $k$, we have $\alpha u_i(\bm{s}) \geq u_i(k,\bm{s_{-i}})$. Therefore, after any group deviation to strategy k (if the final strategy is $\bm{s'}$), we have $u_i(\bm{s'}) \leq u_i(k,\bm{s_{-i}}) + u_i(\bm{s}) \leq (\alpha+1) u_i(\bm{s})$, which gives us the bound. For players whose strategy is other than $k_0$, their utility cannot increase a factor more than $1+\frac{1}{\alpha} \leq (1+\alpha)$. $\blacksquare$

\section{Incentivizing Players} 

So far, we have looked at approximate equilibria for our game, which become fully stable when there are some incentives provided to players or when there is limited inertia in switching. We now consider a more explicit incentivizing technique where a central authority can provide arbitrary payments to any subset of players in order to enforce a desired solution. Ideally, this desired solution is the social optimum although we later consider other high welfare solutions which are easily computable. In such situations, there are often budgets constraints which determine the total such payments that can be made and we wish to obtain bounds on the minimum budget required for any given instance. Our main result here is the complement of Theorem~\ref{theorem_approx}. Informally our result implies that when the solution returned by the Algorithm in Theorem~\ref{theorem_approx} does not have a high social welfare, then we can `stabilize' a solution that is close to $OPT$ with a very small budget. Together, the results indicate that for any instance, we can either compute a good solution of provide small payments to enforce a solution close to OPT.

We now define our incentive scheme: we are interested in `stabilizing' a given solution $\bm{s}$ by providing each player $i$ a payment of $\nu_i OPT$, where $0\leq\nu_i\leq1$, such that her total utility is now $u_i(\bm{s}) + \nu_i \cdot OPT$. More precisely, players are provided this additional utility if and only if they stick to their prescribed strategy under $\bm{s}$. A solution is said to be successfully stabilized if after the additional payments, no player wishes to deviate from strategy $s_i$. Our goal is to bound $\nu$, the total payment that is required as a fraction of the optimum welfare, i.e., $\nu \cdot OPT = \sum_{i}\nu_i \cdot OPT$. Now suppose for that for a given instance of the coordination game, the $\alpha$-Approximate Equilibrium returned by our Algorithm of Theorem~\ref{theorem_approx} has a social welfare of $\rho_\alpha \cdot$OPT for $\rho_\alpha \leq 1$. Note that $\rho_\alpha$ is actually a function of the given instance which we omit as its meaning here is clear. Our second main result is the following. 

\begin{theorem}
\label{theorem_payment}
For any given instance, the optimum solution can be stabilized by providing direct payments to players such that the total payment is at most a fraction $\nu \leq \frac{\rho_\alpha}{(\alpha-1)}$ times the social welfare of OPT.
\end{theorem}
The above result indicates that if we run our algorithm from Theorem~\ref{theorem_approx} and get a solution with low social welfare ($\rho_\alpha$ is small), then we can stabilize the optimum solution with a total payment that is approximately equal to $\rho_\alpha \cdot$OPT. We previously observed that when both $\gamma$ and $m$ are large, the solution returned by our algorithm may not be too efficient. In such cases, we can do much better by implementing the optimal solution and providing small incentives to players.
\begin{proof}
To be concrete, let us assume that we have run the Algorithm of Theorem~\ref{theorem_approx} and obtained a solution whose social welfare is a fraction $\rho_\alpha$ times OPT for that given instance. Let $s^*$ be the social optimum and by definition, $u(s^*) = A(s^*) + P(s^*)$, i.e., the components of the total welfare derived from players' intrinsic preferences and that obtained from player relationships respectively. Let $X \subseteq N$ denote the players to whom payments have to be provided in $s^*$ in order to prevent them from deviating and suppose that $\forall i \in X, k_i$, denotes their best-response strategy. Now how much additional utility $\nu_i$OPT should we provide player $i$? The minimum such incentive required must equal the additional utility that she would gain by performing best-response from OPT, i.e., $v_i OPT = u_i(k_i,s^*_{-i}) - u_i(s^*)$. The second term in the right hand side is bounded by the utility that players derive from the anchors alone. Therefore, we have,
$$\sum_{i \in X} (u_i(k_i,s^*_{-i}) - w_i^{s^*_i}) \geq \sum_{i \in X}(u_i(k_i,s^*_{-i}) - u_i(s^*)).$$
It suffices therefore to say that the term in the LHS of the above inequality serves as an upper bound for the total payments required to stabilize the optimum. Let us overload notation and represent $\sum_{i \in X}u_i(k_i,s^*_{-i}) = A(X) + P(X)$, where as usual the first term represents the utility coming from the intrinsic preferences and the second term the coordination utility.  Now, we have $\nu OPT \leq A(X) + P(X) - \sum_{i \in X}w_i^{s^*_i}$. 

Consider a strategy vector $s'$, where all players choose the same strategy, say $k_0$, for some $1 \leq k_0 \leq m$. We know, $u(OPT) \geq u(s')$. We therefore, have $A(s^*) + P(s^*) \geq A(s') + P(s') \geq P(s') = P_T \geq P(s^*) + P(X)$. The final inequality comes from the fact that $P(X)$ and $P(OPT)$ are disjoint, i.e., the coordination utility that players would derive by deviating from $OPT$ is not a part of $OPT$ itself. But together $P(X) + P(OPT) \leq P_T$, which represents the total coordination utility derived when all players are choosing the same strategy. Removing $P(s^*)$ from both sides, we get, $A(s^*) \geq  P(X)$. Finally we have,
\begin{eqnarray}
\nu OPT &\leq& A(X) + P(X) - \sum_{i \in X}w_i^{s^*_i}\\
& \leq & A(X) + A(s^*) - \sum_{i \in X}w_i^{s^*_i}\\
& = & \sum_{i \in X}w_i^{k_i} + \sum_{i \in N}w_i^{s^*_i} - \sum_{i \in X}w_i^{s^*_i}\\
& = & \sum_{i \in X}w_i^{k_i} + \sum_{i \notin X}w_i^{s^*_i} \\
& \leq & A_T.
\end{eqnarray}
So we now have $A_T$, the sum of the utilities that each selfish player derives from choosing her favorite strategy as an upper bound for $\nu OPT$. We know from the proof of Theorem~\ref{theorem_approx}, that our $\alpha$-Approximate Equilibria have a utility of at least $A_T(\alpha-1)$. So we conclude $\nu OPT \leq \frac{\rho_\alpha}{\alpha-1}$. \end{proof}
Putting this in perspective, if our $1.618$-Approximate Equilibrium has a Social Welfare that is one-tenth of OPT, then we can stabilize the optimum by providing total payments no greater than sixteen percent of OPT. The bottom line for a central enforcer is that a little incentivizing goes a long way in ensuring socially efficient outcomes in coordination games. For every given instance, we can either apply the algorithm of Theorem~\ref{theorem_approx} to compute a good quality approximate equilibrium or enforce even the optimum solution with a small total budget, thereby always ensuring high welfare, stable outcomes. Unfortunately, computing the optimum is NP-hard and therefore, we consider stabilizing other high quality solutions that are a good approximation to $OPT$. In particular, we claim that the $0.8535$ approximation to OPT provided by the polynomial-time algorithm in~\cite{langberg2006approximation} can be stabilized with the same total payments as required for the optimum.
\begin{corollary}
There exists a solution computable in polynomial time whose social welfare is at least $0.8535$ times OPT, which can be stabilized by providing total payments $\nu$ no greater than $\frac{\rho_\alpha}{\alpha-1}$ times OPT.
\end{corollary}
\emph{(Proof Sketch)} We first take the $0.8535$ approximation to the optimum returned by the Algorithm in ~\cite{langberg2006approximation}. Then we consider a solution where all players are present in $k^*$, where $k^*$ is the strategy that maximizes the total intrinsic utility derived by players as we defined previously. Now, let $\bm{s}$ be the solution among these two that has higher social welfare. We claim that $\bm{s}$ can be stabilized with a total payment of $\frac{\rho_\alpha}{\alpha-1}$ times OPT. We do not go over the proof as it is exactly similar to the proof of Theorem~\ref{theorem_payment}. $\blacksquare$

\subsection{Influencing Strong Equilibrium by strengthening relationships}
Previously, we showed that although strong equilibrium may not exist in general, a 2-Approximate strong equilibrium always exists and can be computed. In this section, we look at $\alpha$-Approximate Strong Equilibrium for $1 \leq \alpha \leq 2$ and show that these always exist in the presence of a little additional influence. Our influencing technique is quite different from our previous strategy of providing players direct incentives. Specifically, player $i$ receives full utility from all the other players choosing the same strategy such that $w(i,j)>0$ and a small fraction $\omega$ of the utility from players with whom they originally had $w(i,j) = 0$. From the perspective of viewing our game as a network game, we are essentially creating links between players. We only consider SCGs that obey the correlated coordination condition to eliminate uncertainty over the reward sharing.

\textbf{Model:} We first consider an equivalent formulation of the SCGs that obey the CC condition. Namely, let us associate two parameters with each player $a_i$ and $b_i$ such that when players $i$ and $j$ coordinate, player $i$ receives a utility of $a_ib_jw(i,j)$ and player $j$ receives $a_j b_i w(i,j)$. It is not hard to see that this is at least as general as the CC condition\footnote{given $\bm{\gamma}=(\gamma_1, \cdots, \gamma_N)$ and $w(i,j)$ for all $(i,j)$, we simply replace $w(i,j)$ with ($\gamma_i + \gamma_j$)$w(i,j)$, set $a_i = \gamma_i \; \forall i$ and $b_i=1 \; \forall i$.}. We add an additional constraint over the weights, namely $w(i,j) = 0$ or $w(i,j)=1$ or $w(i,j) = -\infty$. The minus infinity weights can be viewed as conflict edges. In other words, players $i$ and $j$ can never choose the same strategy. Such a model with conflict edges was first considered in~\cite{kleinberg2013information} to study the formation of gossip networks. Given such an SCG, we now define the $\omega$ extension of the game for $\omega \in [\frac{1}{2}, 1]$ to be the social coordination game with the same parameters except $w(i,j) = \omega$ where it was previously zero. 

The incentivizing technique that we use here is quite intuitive. We let players derive full benefit from their neighbors present in the same group, but also give them to a fraction $\omega$ utility from edges where $w(i,j)=0$ originally. Formally, given a SCG, let $N(i)$ be all the $j$ such that $w(i,j)=1$ and $F(i)$ be all the j such that $w(i,j)=0$. After incentivizing, in any given strategy $\bm{s}$, a player's utility is 
$$u_i(\bm{s}) = \sum_{\substack{s_j=s_i \\ j \in N(i)}} a_i b_j + \omega \sum_{\substack{s_j = s_i \\ j \in F(i)}} a_i b_j $$

Our following result can be viewed as an extension of a result in~\cite{kleinberg2013information}, showing the existence of Approximate Strong Equilibria.

\begin{proposition}
The $\omega$-extension of any given SCG obeying the correlated coordination condition with weights as defined above admits a $\frac{1}{\omega}$-Approximate Strong Equilibrium.
\label{theorem_sealphaapproximate}
\end{proposition}

We use the technique of lexicographic ordering that has been used for showing  the existence of strong equilibrium in various guises in literature~\cite{harks2009strong}. The appropriate definition for lexicographic ordering is present in the appendix.
\begin{proof} Consider a $m$-dimensional vector $\pi: A \mapsto \mathbb{R}^m_{+} $ defined on each state $\bm{s}$ such that the $k^{th}$ component of the vector $\pi_k(\bm{s})$ is equal to $\sum_{i: s_i=k}b_i$. Recall that the set of allowed states $A$ only includes strategy vectors where no two conflicted users are present in the same group.

Now consider the solution $\bm{s^*}$ such that all the other solutions are sorted lexicographically smaller than $\bm{s^*}$, i.e. $\pi(\bm{s}) \preceq \pi(\bm{s^*})$ for all $\bm{s}$. We claim that $\bm{s^*}$ is a $\frac{1}{\omega}$-approximate strong equilibrium. Consider any deviation by a group of players from $\bm{s^*}$ to a new solution $\bm{s}$ and let $i$ be the player whose original strategy had the largest weight(smallest index) under the lexicographic ordering. Suppose $i$ deviates from strategy $k$ to $k'$, we claim that $\pi_k(\bm{s^*}) \geq \pi_{k'}(\bm{s})$. If this were not true, then the strategy vector $\bm{s}$ would be lexicographically larger than $\bm{s^*}$. Now the utility that $i$ received originally is
$$u_i(\bm{s^*}) = a_i(\sum_{\substack{s^*_j=s^*_i \\ j \in N(i)}} b_j + \omega \sum_{\substack{s^*_j = s^*_i \\ j \in F(i)}} b_j) \geq \omega a_i\sum_{\substack{j \neq i \\ s^*_j=k}} b_j = a_i\omega (\pi_k(\bm{s^*}) - b_i).$$
$i$'s utility after deviating is
$u_i(\bm{s}) \leq a_i\sum_{s_j = k'}b_j =  a_i(\pi_{k'}(\bm{s}) - b_i) \leq a_i(\pi_{k}(\bm{s^*}) - b_i) \leq \frac{1}{\omega}u_i(\bm{s^*})$. This completes the proof of the theorem.
\end{proof}

\begin{corollary}
When $\omega=1$, every instance of the $\omega$-extension of SCGs obeying with the correlated coordination condition admits a Strong Equilibrium.
\end{corollary}

\section{Social Coordination with Complementarities}
In this section, we generalize our model of social coordination to include scenarios where the benefits due to coordination may exhibit complementarities or supermodularities. So far, we have looked at incentivizing players and computing approximate equilibrium for a very specific model of social coordination where the benefits are additively separable. However, it is not hard to see that most of our results extend to the case when utilities and welfare may be submodular in the set of players. We now generalize this further and consider SCGs wherein the benefits gained by coordinating with a set of players can be much larger than the utility gained by individually coordinating with each of them. 

\textbf{Model for Generalized Social Coordination Game.} We consider a game with players belonging to a set $N$, a strategy set $M$ and functions $w:2^{N \cup M} \mapsto \mathbb{R^+}$ and $u_i:2^{N \cup M} \mapsto \mathbb{R^+}$ for all players $i \in N$. Players still choose a single strategy and in any given state of the game $\bm{s}$, welfare due to all players $X_k$ choosing strategy $k$ is given by $w(k+X_k)$. This welfare is then distributed among the players, so we have
$$\displaystyle \sum_{i \in X_k}u_i(\bm{s}) = \sum_{i \in X_k}u_i(k+X_k) = w(k+X_k).$$ The total social welfare of the solution is given by $u(\bm{s}) = \sum_i u_i(\bm{s}) = \sum_{k \in M}w(k+X_k)$. Observe that we have made no assumptions on the functions $w$ and $u_i$. Therefore, we may have that $w(S \cup T) \geq w(S) + w(T)$, which is not possible in our previous model. We only require that the total welfare generated is divided among the players.

There has always been considerable interest in utility-maximization games where either the resources or players exhibit complementarities. Situations of interest include combinatorial auctions wherein valuation functions can be supermodular in the set of items received, or group formation where people may possess specific sets of skills. However, there has not been much progress in this realm of algorithmic game theory as even the fundamental algorithmic problem of maximizing welfare in the presence of complements is considered \textit{notoriously difficult}. Recent papers~\cite{abraham2012combinatorial, feige2013welfare, lehmann2001combinatorial} have sought to overcome this obstacle by limiting the amount of supermodularity (or complementarity) present in the valuations in combinatorial auctions. With regard to full-information games such as ours, most similar to our model is the welfare-sharing game proposed by Bachrach et al.~\cite{bachrach2013strong} who showed bounds on the Strong Price of Anarchy. Instead, here we focus on questions related to existence and computation of stable solutions.

We begin by showing that there exist instances of the Generalized SCG (with complementarities) that exhibit no stable solutions when the criterion considered is (Approximate) Nash Stability. Previously, we proved that all instances of SCGs without complementarity admit a $1.618$-Approximate Equilibrium. However, when the player relationships become complementary, there exist instances which do not admit an $\alpha$-Approximate Equilibrium for any $\alpha \geq 1$.

\begin{proposition}
For any $c \geq 1$, there exists an instance of the Generalized Social Coordination Game which does not admit an $\alpha$-Approximate Equilibrium for all $1 \leq \alpha \leq c$.
\end{proposition}
\begin{proof}
The example that we use here is an extension of the triangle example from Proposition~\ref{prop:nonexistence}. Once again consider three players $i,j,k$ and three strategies $1,2,3$. We define the utility function for player $i$ for different scenarios. The functions for players $j$ and $k$ are defined similarly. $u_i(1) = u_i(1+k)= c^2, u_i(1+j) = u_i(1+j+k)=c^3, u_i(2) = c, u_i(2+j) = u_i(2+j+k)=c^3, u_i(3)=u_i(3+k)=0, u_i(3+j+k)=c$. Once again, it is not hard to verify that in any state at least one player can deviate and improve her utility by a factor $c$. Therefore, this instance admits no $\alpha$-Approximate Equilibrium for any $\alpha \leq c$.
\end{proof}

The non-existence of Approximate Nash Equilibrium for Generalized SCGs forces us to consider a more restricted model. We consider two forms of limited complementarities that have also been considered in other settings in literature. First, we look at hypergraph SCGs, a direct extension of the Network Game model mentioned in Section~\ref{sec:prelim}. We show some sufficient conditions for the existence of a Nash Equilibrium. We then move on to a model where the degree of supermodularity is bounded by some constant $r$. We adapt our \onebr algorithm to obtain a $\Theta(r)$-Approximate equilibrium.

\subsection{Hypergraph Social Coordination Games}
Once again, we view our SCG as a network game where the nodes include the players and the anchored nodes (strategies). There exist a set of hyperedges $E$, and each hyperedge has a weight $w_e$. The total welfare of a set $S$ is given by $w(S) = \sum_{e \in S}w_e$, where the summation is over the hyperedges whose member vertices all belong to $S$. The welfare due a hyperedge is distributed among the vertices that comprise the hyperedge. For each player $i \in e$, we have a constant $\gamma^e_{i}$ such that $i$'s utility due to the hyperedge is given by $\gamma^e_{i}w_e$. As is natural, the constants $\gamma^e_i$ satisfy $\sum_{i \in e}\gamma^e_i =1$. Such a model with a hypergraph induced on a set of items (instead of the players) was considered in~\cite{abraham2012combinatorial} and they provided a $r$-Approximation Algorithm for welfare maximization where $r$ is the maximum number of vertices present in a single hyperedge. We now extend our previous Correlated Coordination Condition for this Generalized SCG and show that Nash Equilibrium always exists for instances that obey the CC condition.

\textbf{(Correlated Coordination Condition)} A given instance of the SCG on hypergraphs is said to satisfy this condition if $\exists$ a vector $\bm{\gamma}=(\gamma_1, \cdots, \gamma_N)$ such that $\forall e \in E$ with $w_e >0$, we have $\gamma^e_{i}=\frac{\gamma_i}{\sum_{j \in e}\gamma_j}$ for all the vertices that belong to the hyperedge.

\begin{proposition}
Hypergraph Social Coordination Games which obey the Correlated Coordination (CC) Condition admit a potential function. Therefore, best-response dynamics always converge to a Nash Equilibrium in such games. 
\end{proposition}
\begin{proof}
Given a social coordination game with weights $w_e, \forall e \in E$ and $\bm{\gamma}=(\gamma_1, \cdots, \gamma_N)$, that obeys the CC condition, we claim that the following is an inexact potential function for this game.
$$\Phi(\bm{s}) = \displaystyle \sum_{\substack{e:\\s_i=s_j \forall i \in e, j \in e}}\frac{w_e}{\sum_{j \in e}\gamma_j}.$$

Note that $\gamma_i = 0$ for all the anchored nodes as they do not receive any utility. The proof follows along the lines of the proof of Theorem~\ref{theorem_ccexistence}
\end{proof}

\subsection{Social Coordination Games with bounded supermodularity}
We consider general SCGs where the welfare and utility may exhibit complementarities but we bound the amount of complementarity in each player's utility by a multiplicative factor $r \geq 1$. Specifically, we define a $r$-Supermodular SCG as a game where for all players $i$, we have $u_i(S \cup T) \leq r(u_i(S) + u_i(T))$ for any $S$ and $T$. Informally this implies that even though the welfare due to a group of people coordinating is more than the sum of welfare due to smaller subgroups coordinating, the additional welfare due to complementarity is bounded by some parameter. The Social Coordination Games considered in Sections 1-5 satisfy this criterion with $r=1$ and therefore, $r$-Supermodular SCGs are natural generalizations of linear or submodular Social Coordination Games.

Similar notions of restricted complementarities were studied in~\cite{lehmann2001combinatorial} and ~\cite{feige2013welfare}. Both papers considered the optimization problem of welfare maximization and obtained approximation factors linear in the degree of supermodularity. We now consider the problem of computing an approximate equilibrium for SCGs with bounded supermodularity and show that every $r$-Supermodular SCG admits an approximate equilibrium whose stability factor is linear in the parameter $r$. It is an interesting open question whether this bound is tight or whether we can obtain stability factors sublinear in $r$.

Once again we turn to our \onebr algorithm and adjust its parameter to obtain an $r+\epsilon$-Approximate Equilibrium for some $\epsilon <1$.

\begin{theorem}
For any given instance of a $r$-Supermodular SCG, $\exists$ an $\epsilon < 1$, such that running the \onebr algorithm for $\alpha=r+\epsilon$ returns a $r+\epsilon$-Approximate Nash Equilibrium.
\end{theorem}
In other words, we can always compute a $\Theta(r)$-Approximate Nash Equilibrium for any $r$-Supermodular SCG.
\begin{proof}
Let $k_0$ be the starting strategy for all players. Once again we assume that players deviate from $k_0$ in rounds, let $X_k(t)$ refer to the set of players whose strategy is $k$ after $t$ rounds.  Suppose we allow players to deviate only when it improves their utility by a factor $\alpha$ or more (we will see what this factor is shortly). Let $\bm{s}$ be the final solution. For every player $i$ in $\bm{s}$ whose strategy is $k_0$, she cannot improve her utility by a factor more than $\alpha$ if she deviates (by definition). For any player who deviated to a strategy $k \neq k_0$ at some time $t$, if her best-response is $k'$, we then have
$$u_i(k',s_{-i}) = u_i(X_k') \leq u_i(X_k'(t) \cup X_{k_0}(t)) \leq r(u_i(X_k'(t)+u_i(X_{k_0}(t))).$$

But since $k$ was $i$'s best-response at time $t$, we have $u_i(\bm{s}) \geq u_i(X_k'(t)$ and $u_i(\bm{s}) \geq \alpha u_i(X_{k_0}(t))$. So we can finally conclude that 
$$u_i(k',s_{-i}) \leq r(u_i(X_k'(t)+u_i(X_{k_0}(t))) \leq ru_i(\bm{s})(1+\frac{1}{\alpha}).$$
The stability factors for the two sets of players are $\alpha$ and $r(1+\frac{1}{\alpha})$ respectively. In the best possible case, they are both equal, giving us $\alpha = \frac{1}{2}\displaystyle(r + \sqrt{r(r+4)}) < r+1$ for $r \geq 1$.
\end{proof}

\section*{Acknowledgements}
This work was supported in part by NSF awards CCF-0914782, CCF-1101495, CNS-1017932, and CNS-1218374.

\bibliography{bibliography}
\bibliographystyle{plain}
\newpage

\appendix
\section*{Appendix}
\renewcommand{\thesubsection}{\Alph{subsection}}
\subsection{Proofs from Section 3}
\begin{prop_app}{prelim_poa}
The Price of Anarchy (PoA) for the Social Coordination Game over all instances is at most $m$.
\end{prop_app}
\begin{proof}
\begin{definition}
A payoff-maximization game is $(\lambda, \mu)$-semi-smooth if there exists a mixed strategy $\sigma_i$ for each player $i$ such that for every outcome $s$
\begin{equation*}
\mathbf{E}_\sigma\left[\sum_{i \in N} u_{i}(\sigma_i, s_{-i}) \right] \geq \lambda \cdot u(\bm{s}^{*}) - \mu \cdot u(\bm{s}).
\end{equation*}
\end{definition}
It is easy to show that a game being semi-smooth implies all the same results from \cite{roughgarden2009intrinsic} as smoothness, including a bound of $\frac{1+\mu}{\lambda}$ for the price of total anarchy. We now show that our Social Coordination Game is $(\frac{1}{m},0)$-semi-smooth for all instances. This immediately gives us an upper bound of $m$ for the Price of Anarchy not just for Nash Equilibrium but for other more general solution concepts as well. 

For any given strategy vector $\bm{s}$, for each player $i$, let $\sigma_i$ be the mixed strategy where the player chooses each available strategy with equal probability $\frac{1}{m}$. Then we have,
\begin{align*}
\mathbf{E}\left[\sum_{i \in N} u_{i}(\sigma_{i}, s_{-i}) \right]
= \sum_{1\leq k\leq m}\frac{1}{m}(w_i^k + \sum_{s_j=k} \gamma_{ij}w(i,j))
= \frac{1}{m}(\sum_{1\leq k}w_i^k + \sum_{j \in N}\gamma_{ij}w(i,j)).
\end{align*}
The final inside the parenthesis on the right hand side represents the maximum utility that player $i$ could receive in any given solution and thereby acts as an upper bound for her utility in the optimum solution. Summing the above equation for all players, we have, $\sum_{i \in N}\mathbf{E}\left[u_{i}(\sigma_i, s_{-i}) \right] \geq \frac{1}{m} \cdot u(\bm{s}^{*})$. We therefore conclude that our Social Coordination Game is $(\frac{1}{m},0)$-semi-smooth. 
\end{proof}

\begin{prop_app}{pos_symmetric}
The price of stability is at most $(2-\frac{1}{m})$ in social coordination games with symmetric relationships, i.e. $\gamma=1$. Further, this bound is tight.
\end{prop_app}
\begin{proof}

\begin{definition}
A strategy vector $\bm{s^0}$ is said to satisfy the \textbf{Minimum Intrinsic Preference condition}(MIP) if the sum of utilities derived by players purely from their intrinsic preference for their respective strategy ($w_i^{s^0_i}$) is at least $\frac{1}{m}$ times the utility players derive from the strategies in any other solution.
\end{definition}
For a strategy vector $\bm{s^0}$ that satisfies the MIP condition, we are guaranteed that $A(\bm{s^0}) \geq \frac{1}{m} A(\bm{s})$ for any given strategy vector $\bm{s}$.

\begin{proposition}
\label{macbr}
Let $\bm{s^0}$ be a solution that satisfies the Minimum Intrinsic Preference condition. Then, any best-response dynamics starting from $\bm{s^0}$ results in a Nash equilibrium $\bm{s}$ such that the social welfare of $\bm{s^0}$ is at most $(2-\frac{1}{m})$ times that of $\bm{s}$.
\end{proposition}

What this implies is that best-response dynamics from certain ``good" starting points is guaranteed to terminate in Nash equilibria without much loss in social welfare. If the social optimum satisfies the MIP condition, then Proposition~\ref{macbr} tells us that there exists a Nash equilibrium whose social welfare is at least $\frac{1}{(2-\frac{1}{m})}$ times that of the optimum. We first prove the proposition and then show that the social optimum does indeed satisfy the MIP condition, which immediately proves Proposition~\ref{pos_symmetric}.

\noindent\emph{(Proof of Proposition~\ref{macbr})}\\
Let $\bm{s}$ be the Nash equilibrium in which the best-response dynamics from a solution $\bm{s^0}$ that obeys the MIP condition terminates. Since the value of the potential function cannot decrease after any best-response sequence, we have $\Phi(\bm{s}) > \Phi(\bm{s^0})$. This implies,
$$ u(\bm{s})= 2\Phi(\bm{s})-A(s)>2\Phi(\bm{s^0})-A(\bm{s}) = u(\bm{s^0})+A(\bm{s^0})-A(\bm{s}).$$Since $\bm{s^0}$ satisfies the MIP condition, we know that $A(\bm{s^0}) \geq \frac{1}{m} A(s)$. This gives us \begin{eqnarray}
& u(\bm{s}) & > u(\bm{s^0}) - (1-\frac{1}{m})A(\bm{s})\\
\implies & (2-\frac{1}{m})u(\bm{s})&> u(\bm{s^0}),
\end{eqnarray}
where the last step comes from the fact that $A(\bm{s}) \leq u(\bm{s})$. This completes the proof of the claim.

It is easy to see that $OPT$ does satisfy the MIP condition. First, we construct a solution $\bm{s_2}$ where all the players choose the same strategy, the one maximizing total utility. Since, there exist only $m$ strategies, we know that the utility players derive from that strategy is at least $\frac{1}{m}$ times that of the sum of utilities all players derive from all strategies, which is in turn greater than $A_T$, i.e., we have $u(\bm{s_2}) \geq \frac{1}{m}A_T + P_T$. Since $u(OPT) \geq u(\bm{s_2})$, we have $A(OPT) + P(OPT) \geq \frac{1}{m}A_T + P_T$. Note that $P(OPT)$ can never be greater than $P_T$ since $P_T$ includes all the edges between players in a given graph. This means that $A(OPT)$ has to be at least $\frac{1}{m}A_T$ for the inequality to hold. Now for any given solution $\bm{s'}$, $A(\bm{s'}) \leq A_T$, which finally gives us the desired result.
$$A(OPT) \geq \frac{1}{m}A_T \geq \frac{1}{m}A(\bm{s'}), $$ for all $\bm{s'}$.

We note here that this gap between the social optimum and the stable solution maximizing welfare occurs due to the fact that edges between players count only once in the potential function but twice towards the social welfare. The worst case does indeed occur when players choose to follow their most preferred strategy selfishly as opposed to being with their neighbors. We now give a family of examples that illustrate this inherent inefficiency that comes with stability and in the process show that the bound of $2-\frac{1}{m}$ is tight in the limiting case. 

\textbf{Example} Consider an instance with $m$ strategies and $m$ players where each player has a different preferred strategy from which she derives a utility of $w_i^{i} = r + \epsilon$. There is one player $p$ such that $w(i,p)=r$ for all other players $i$. The only Nash equilibrium has a perfectly fragmented society where each player chooses her preferred strategy, achieving a total welfare of $mr + m\epsilon$, whereas all players choose the same strategy in the social optimum which has a welfare of $(2m-1)r + \epsilon$. In the limiting case, as $\epsilon \to 0$, the ratio between the welfares approaches $2-\frac{1}{m}$.

\end{proof}
\subsection{Proofs from Section 4}
\begin{prop_app}{prop:1414approx}
The following algorithm returns a $\sqrt{2}$-Approximate Equilibrium for all instances with three strategies and no algorithm can give a better approximation factor over all instances as this ratio of $1.414$ is tight.
\end{prop_app}
\begin{enumerate}
\item Run \emph{Algorithm 1} for $m=2$ (for say strategies $1$ and $2$) and all players so that no player wants to deviate from one of those two strategies to the other. 
\item While there exists a player, who can deviate from either strategy $1$ or strategy $2$ and improve her utility by a factor $\sqrt{2}$, allow her to deviate to strategy $3$. Once that player deviates, run the algorithm for $m=2$ (for strategies $1$ and $2$) so that players choosing these two strategies are stable with respect to each other.
\item Repeat step 2 until no player can to deviate to 3 and increase her utility by a factor $\sqrt{2}$.
\end{enumerate}
\begin{proof}
Let the final solution be $\bm{s}$. We consider players choosing strategy $3$ in rounds. Suppose $i$ is the player choosing $3$ at some time (round) $t$, then let $s_{-i}^t$ be the strategy vector at time $t$ indicating the strategies of all players other than $i$. Let $X_k(\bm{s})$ be the set of players whose strategy is $k$ under $\bm{s}$. We have by definition, 

\begin{align*}
u_i(3,s_{-i}^t) & \geq  \sqrt{2}u_i(2,s_{-i}^t)\\
u_i(3,s_{-i}^t) & \geq  \sqrt{2}u_i(1,s_{-i}^t)
\end{align*}
In other words, player $i$ receives more utility by a factor $\sqrt{2}$ from strategy 3 than she does from either $1$ or $2$ at that time. Now observe that in the final solution, $X_1(\bm{s}) \cup X_2(\bm{s})$ cannot be more than 
the set of players present in strategies $1$ and $2$ at any time $t$ before termination. Therefore, for player $i$, both $u_i(1,\bm{s_{-i}})$ and $u_i(2,(\bm{s_{-i}}))$ are bounded from above by $u_i(1,s_{-i}^t) + u_i(2,s_{-i}^t)$. So we finally have the maximum utility $i$ can get by switching from strategy 3 to either 1 or 2 in $\bm{s}$,
$$u_i(1,s_{-i}^t) + u_i(2,s_{-i}^t) \leq \frac{2u_i(3,s_{-i}^t)}{\sqrt{2}}) \leq \sqrt{2}u_i(\bm{s})).$$

The last inequality comes from the fact that since the set of players in strategy $3$ only increases in size, $i$'s final utility is at least as much as the utility she received when she deviated to 3. By definition, the algorithm terminates when players in strategies $1$ and $2$ cannot improve their utility by more than a factor $\sqrt{2}$ by deviating to $3$, so the stability factor is trivially applicable for such players. Finally, players in strategy $1$ 
do not wish to deviate to $2$ and vice-versa by definition. This completes our proof.
\end{proof}

\begin{lem_app}{1shotbr_swanchor}
The \emph{1-shot $\alpha$-BR} algorithm, beginning with any starting strategy $k_0$, returns a solution whose social welfare is at least $\displaystyle\frac{A_T}{\alpha}$
\end{lem_app}
\begin{proof}
We use the same notation as before. Note that any player who performs her best-response gets utility at least $best(i)=\max_k{w^k_i}$, since she could always deviate to (or stay in) her preferred strategy. Her final utility is also at least $best(i)$ as player utilities do not decrease when more people choose the same strategy. Any player whose final strategy is $k_0$ cannot deviate and improve her utility by a factor more than $\alpha$, by definition. Therefore, her utility has to be at least $\frac{best(i)}{\alpha}$. Suppose the final strategy vector is $\bm{s}$, we then have $u(\bm{s}) = \sum_{i} u_i(\bm{s}) \geq  \sum_{i} \frac{best(i)}{\alpha} \geq \frac{A_T}{\alpha}$. \end{proof}

\begin{lem_app}{lemma_asymmetricsw}
The \emph{one-shot $\alpha$-BR} algorithm with $\bm{s_0}$ as the starting state (where all players choose a strategy $k_0$), results in a social welfare of at least $A(\bm{s_0}) + \frac{\alpha}{\gamma+1}P(s_0)$, where $\gamma$ is the Maximum Relationship imbalance.
\end{lem_app}
\begin{proof}
Consider a deviation from some intermediate state $\bm{s'}$ by player $i$, which improves her utility by a factor $\alpha$. Then, the original social welfare decreases by (not counting the new utility yet) at most $w_i^{s'_{i}} + \sum_{s'_j=s'_i}(\gamma_{ij}+\gamma_{ji})w(i,j)$. But since, by definition of $\gamma$, we have $\gamma_{ji} \leq \gamma.\gamma_{ij}$, we can now bound the decrease in utility by $w_i^{s'_{i}} + \sum_{s'_j=s'_i}(\gamma_{ij}+\gamma.\gamma_{ij})w(i,j) = w_i^{s'_{i}} + (\gamma+1)\sum_{s'_j=s'_i}\gamma_{ij}w(i,j)$. The increase in social welfare is at least the new utility of player $i$ which is a factor $\alpha$ times her old utility. Therefore this increase is at least, $\displaystyle \alpha(w_i^{s'_{i}} + \sum_{s'_j=s'_i}\gamma_{ij}w(i,j)).$

Therefore, in any deviation, for every unit of utility lost from a player's intrinsic preference, $\alpha$ units of utility are gained and for $(\gamma+1)$ units of utility from player coordination lost, $\alpha$ units are gained. Also observe that once a player deviates to a strategy, she cannot change after that, so players' utilities never decrease. It is not hard to verify that the maximum social welfare is lost when all the welfare lost comes from the 
player coordinations before a player switched from $k_0$. Since the maximum utility due to coordinations in $\bm{s_0}$ is $P(\bm{s_0})$, the resulting utility gain from the deviations is at least $\frac{\alpha P(\bm{s_0})}{\gamma +1}$. In the worse case, the utility due to intrinsic preference that players in $\bm{s_0}$ is retained in $\bm{s}$, giving us the desired social welfare of $A(\bm{s_0}) + \frac{\alpha}{\gamma+1}P(\bm{s_0})$. 
\end{proof}

\begin{thm_app}{theorem_approx}
The following algorithm returns an $\alpha$-Approximate Equilibrium for $\alpha \in [\phi, 2]$ whose social welfare is at least a fraction $\displaystyle \frac{\alpha - 1}{1 + \frac{(\gamma+1)}{\alpha}(\alpha - (1 + \frac{1}{m}))}$ of the optimum when $(\gamma+1) \leq m\alpha$, and $\max(\frac{\alpha}{\gamma+1}, \frac{\alpha-1}{(m-1) + (\alpha-1)})$ times OPT otherwise.
\end{thm_app}

\begin{quote}
\textbf{Algorithm:} ``For a given $\alpha$, run \emph{One-shot $\alpha$-BR} and \emph{One-shot $\frac{1}{\alpha-1}$-BR} with $k^*$ as the starting strategy. Let the returned solutions be $\bm{s_1}$ and $\bm{s_2}$. Return the solution among these two with greater social welfare."
\end{quote}

\begin{proof}
We begin by establishing that both the procedures in the algorithm result in a $\alpha$-Approximate Equilibrium. Applying Lemma~\ref{1shotbr_stabilityfactor} to the solution returned by the \emph{1-shot $\frac{1}{\alpha-1}$-BR} procedure, we get that its stability factor is $\max(\frac{1}{\alpha-1}, \frac{1}{\frac{1}{\alpha-1} }+1)$, that is $\max(\frac{1}{\alpha-1}, \alpha)$. $\bm{s_1}$ is, therefore, an $\alpha$-approximate Nash equilibrium as $\alpha \geq \frac{1}{\alpha-1}$ for $\alpha \in [\phi,2]$. Similarly, applying the lemma to $\bm{s_2}$, we get that its stability factor is $\max(\alpha, \frac{1}{\alpha}+1)$, which is also $\alpha$ for the given range. Since both $\bm{s_1}$ and $\bm{s_2}$ are $\alpha$-Approximate Equilibrium, one can make the same conclusion for $\bm{s}$, since it is the best of the two solutions.

First, we claim that the total social welfare of the solution where all players choose the strategy $k^*$ is at least $\frac{A_T}{m} + P_T$. Recall that $A_T$ is the utility when each player chooses her favorite strategy. Since the total intrinsic utility derived when players choose $k^*$ is at least as much as the intrinsic utility from any other strategy and since there are only a total of $m$ strategies, we conclude that $A(\bm{s_0})=\frac{A_T}{m}$, where $\bm{s_0}$ is the strategy under which all players choose $k^*$.  $P(\bm{s_0})=P_T$ since all players are present in the same strategy. Now, applying Lemma~\ref{1shotbr_swanchor} to our algorithm's $\alpha$-BR procedure (where the starting state is $k^*$), we get that the final social welfare of the solution, $u(\bm{s})$, is at least $\frac{A_T}{m} + \frac{\alpha}{\gamma+1}P_T$. Applying Lemma~\ref{1shotbr_swanchor} to our \emph{One-shot $\frac{1}{\alpha-1}$-BR} algorithm, we get that our computed solution has a social welfare of at least $(\alpha-1)A_T$. Since our hybrid algorithm returns the best of these two solutions, we have

$$u(s) \geq \max((\alpha-1)A_T, \frac{A_T}{m} + \frac{\alpha}{\gamma+1}P_T).$$ Let us define $A_T = y(A_T + P_T)$, for some $y \leq 1$. Then, we get,
\begin{equation}
\label{eqn_socialwelfareasymmetric}
u(\bm{s}) \geq \max\left((\alpha-1)y,  \frac{\alpha}{\gamma+1} +  y(\frac{1}{m} - \frac{\alpha}{\gamma +1})\right)OPT
\end{equation}
First, let's assume that $\frac{\alpha}{\gamma+1} \geq \frac{1}{m}$. Then, observe that the first term in the parenthesis is increasing in $y$ and the second term is decreasing. So, the maximum of the two terms is minimized when they are equal, i.e., $(\alpha-1)y = \frac{\alpha}{\gamma+1} +  y(\frac{1}{m} - \frac{\alpha}{\gamma +1})$ or $y = \frac{1}{1 + \frac{(\gamma+1)}{\alpha}(\alpha - (1 + \frac{1}{m}))}$. So we have $u(\bm{s})$ is always greater than $y (\alpha-1)OPT = \frac{OPT}{1 + \frac{(\gamma+1)}{\alpha}(\alpha - (1 + \frac{1}{m}))}$.

Now, for the case when $\frac{\alpha}{\gamma+1} < \frac{1}{m}$, our results show that $u(\bm{s}) \geq \max(\frac{\alpha}{\gamma+1}, \frac{\alpha-1}{(m-1) + (\alpha-1)})$. The first term simply follows from equation~\ref{eqn_socialwelfareasymmetric} from which we note that when $\frac{\alpha}{\gamma+1} < \frac{1}{m}$, $u(\bm{s}) \geq \frac{\alpha}{\gamma+1}.$ In order to show the second term, we prove a general result that is applicable for all instances, namely that the social welfare is bounded from below by $\approx \frac{1}{m}$.
%

\begin{lemma}
\label{asymmetric_swlowerbound}
The \emph{one-shot $\frac{1}{\alpha-1}$-BR} algorithm beginning from any starting state results in a social welfare which is at least a fraction $\frac{\alpha -1}{(m-1) + (\alpha -1)}$ times OPT.
\end{lemma}
\begin{proof}Consider the outcome of \emph{One-Shot $\frac{1}{\alpha-1}$}-BR algorithm. Any player who deviated from their initial strategy received at least a fraction $\frac{1}{m}$ times their utility in OPT as they performed best-response and their utility cannot decrease after that round. Players still present in the initial state $k_0$ in $\bm{s}$ receive at least $\frac{\alpha-1}{(m-1)+(\alpha-1)}$ times their utility in $OPT$. To show this, observe that for any such $i$, $u_i(\bm{s}) \geq (\alpha-1)u_i(s'_i,s_{-i})$ for $s'_i \neq s_i$ and $u_i(\bm{s}) = u_i(s_i,s_{-i})$. Summing up these inequalities and the last equality, we have $\displaystyle ((m-1)+(\alpha-1))u_i(\bm{s}) \geq (\alpha-1)u_i(OPT)$, which gives our desired fraction. Now since every player is receiving at least the fraction $\frac{\alpha-1}{(m-1)+(\alpha-1)}$ of her utility in $OPT$, we have that the social of our solution is also the same fraction of the optimum solution welfare. This gives a lower bound on the social welfare of the solution returned by our algorithm for all cases and completes the proof of the theorem and sub-lemma.
\end{proof}
\end{proof}

\begin{prop_app}{claim_seexistence}
There are instances of the symmetric Social Coordination Game with $m=3$ groups where Strong equilibrium does not exist. However, whenever Strong equilibrium exists, its social welfare is no less than half of OPT when player relationships are symmetric.
\end{prop_app}
\begin{proof}
The second part of the claim also applies for the symmetric SCG and technique we use is quite common in the Strong Price of Anarchy literature. Let $\bm{s}$ be any Strong Nash equilibrium(SNE) and $\bm{s^*}$ denote the optimal solution. Consider a group deviation from $\bm{s}$ by the set of players whose strategy is $k$ in OPT ($X_k(\bm{s^*})$). Suppose these players deviate to strategy $k$ from $\bm{s}$ forming a new strategy vector $\bm{s'}$. Since $\bm{s}$ is a SNE, in any group deviation there must be at least one player whose utility cannot increase. Let this player be $i$. Then, $i$ must have received at least as much utility in $\bm{s}$ as she does after deviating to $k$ along with the players in $X_k(\bm{s^*})$. This means that $u_i(\bm{s}) \geq u_i(\bm{s'}) \geq u_i(\bm{s^*})$, where the last inequality comes from the fact that in $\bm{s'}$, $X_k(\bm{s^*}) \subseteq X_k(\bm{s'})$.

Now consider a deviation by all players in $X_k(\bm{s^*})$ to strategy $k$ other than player $i$. Once again, $\exists$ a player $j$ such that $j$ receives more utility from strategy $s_j$ than from the set of players in $X_k(\bm{s^*}) - \{ i \}$. So we get the inequality $u_j(\bm{s}) \geq u_j(\bm{s^*}) - \frac{w(i,j)}{2}$, since $\gamma_{ij}=\gamma_{ji}=1$. If we again look at a deviation to strategy $k$ by all players other than $i$ and $j$, there must be one player whose utility after this deviation is not more than the utility before deviating in $\bm{s}$, giving us a similar inequality as above. We therefore, repeat this process considering one less player each time, and sum over all the inequalities obtained. It is not hard to see that in the left hand side, we get $\sum_{i \in X_k(\bm{s^*})} u_i(\bm{s})$ and the right hand side $\sum_{i;s^*_i=k}w_i^{k} + 0.5\sum_{s^*_i=s^*_j=k}w(i,j)$, which is at least $\sum_{i \in X_k(\bm{s^*})} \frac{u_i(\bm{s^*})}{2}$. This is because each edge among the players of $X_k(\bm{s^*})$ gets summed up only once (when one of the players leave). We now repeat process for all strategies $k$ and set of players $X_k(\bm{s})$, and sum over all the sets

Now summing these inequalities over all players and all strategies, we get
$$\sum_{i} u_i(\bm{s}) \geq \sum_{i} \frac{u_i(\bm{s^*)}}{2},$$ which gives us our desired bound on the Strong Price of Anarchy.
\end{proof}

\begin{prop_app}{prop_se2app}
For $\alpha \geq 1$, the \emph{One-shot $\alpha$-BR} algorithm returns a $(1+\alpha)$-Approximate Strong Equilibrium.
\end{prop_app}
\begin{proof}

We already know by Lemma~\ref{1shotbr_stabilityfactor} that any deviation by a single player cannot increase her utility by a factor greater than $\alpha$. It is not immediately clear why this factor holds for group deviations. As mentioned previously, let $V_1$ refer to the players who deviated from the initial state $k^*$ and $V_2$ be the players whose strategies under $\bm{s}$ are $k^*$. We first claim that in any deviation by a set of players $S$ that includes at least one player from $V_1$, the first player in $V_1 \cap S$ to have originally deviated from $k^*$ cannot improve her utility by more than a factor $(1+\frac{1}{\alpha})$.

Consider any group deviation by players $S \subseteq V$ which contains at least one player from $V_1$. Let the new strategy vector be $\bm{s_2}$. Let $i$ be the player in $S \cap V_1$ who deviated from $k^*$ before all the other players in $S$, say at time $t$. Since all the other players in $S$ had $k^*$ as their strategy at time $t$, they must have belonged to $X^t_{k^*}$. Suppose $i$'s strategy in $\bm{s_2}$ is $k$, then the utility that $i$ gets after deviating is at most $u_i(k,s^t_{-i}) + u_i(k^*,s^t_{-i})$. But this is not greater than $u_i(s_i,s^t_{-i}) + \frac{1}{\alpha}u_i(s_i,s^t_{-i})$, where the first term appears because $s_i$ is $i$'s best response and the second is because the new utility after deviating is at least $\alpha$ times the old utility which was $u_i(k^*,s^t_{-i})$. Therefore, we have $u_{i}(\bm{s_2}) \leq (1+\frac{1}{\alpha})u_{i}(\bm{s})$.

Now suppose $S \subseteq V_2$, i.e. $S \cap V_1 = \emptyset$. We know
\begin{equation}
\label{equation_v2utility}
u_i(\bm{s_2}) \leq u_i(k,s_{-i}) + u_i(k^*,s_{-i})
\end{equation}
for some $i \in S$, whose strategy under $\bm{s_2}$ is  $k$. This is true because the all the new players present in strategy $k$ under $\bm{s_2}$ must have all been players whose strategy in $\bm{s}$ was $k^*$. But $\forall i \in V_2$, we know that the utility these players originally received when their strategy was $k^*$ could not have improved by more than a factor $\alpha$ after deviation, if they had deviated individually. Applying this observation to equation~\ref{equation_v2utility}, we have $u_i(\bm{s_2}) \leq u_i(k,s_{-i}) + u_i(k^*,s_{-i}) \leq \alpha u_i(\bm{s}) + u_i(\bm{s})$. Therefore, in any group deviation from $\bm{s}$, there is at least one player whose utility increases by a factor no more than $(1+\alpha)$ for the solution returned by the \emph{One-shot $\alpha$-BR} algorithm. 

Specifically for $\alpha=1$, we get a $2$-Approximate Strong Equilibrium for every instance. Using the Sub-Lemmas of Theorem~\ref{theorem_approx}, we immediately get that this solution gives us a guarantee of at least half the social welfare of Theorem~\ref{theorem_approx}. 
\end{proof}

\subsection{Proofs from Section 5}

\subsubsection*{Lexicographic Ordering}

We borrow the following definition from~\cite{harks2009strong}.

\begin{definition}
Let $\pi(\bm{s_1})$ and $\pi(\bm{s_2})$ be the potential vector for two states $\bm{s_1}, \bm{s_2}$ and denote by $\tilde{\pi}(\bm{s_1}), \tilde{\pi}(\bm{s_2})$ be the sorted potential vectors respectively,  derived by sorting the entries of the actual vector in a non-increasing order, that is $\tilde{\pi}_1(\bm{s_1}) \geq  \cdots \geq \tilde{\pi}_m(\bm{s_1})$ and similarly for $\tilde{\pi}(\bm{s_2})$. Then $\pi(\bm{s_1})$ is strictly sorted lexicographically smaller than $\pi(\bm{s_2})$ ($\pi(\bm{s_1}) \prec \pi(\bm{s_2}))$ if $\exists$ an index $k$ such that $\tilde{\pi}_l(\bm{s_1}) = \tilde{\pi}_l(\bm{s_2})$ $\forall l < k$ and $\tilde{\pi}_k(\bm{s_1}) < \tilde{\pi}_k(\bm{s_2})$. Similarly, we say $\pi(\bm{s_1})$ is sorted lexicographically smaller than $\pi(\bm{s_2})$ ($\pi(\bm{s_1}) \preceq \pi(\bm{s_2}))$ when either $\pi(\bm{s_1}) \prec \pi(\bm{s_2})$ or $\pi(\bm{s_1}) = \pi(\bm{s_2})$.
\end{definition}

\end{document}